\documentclass[preprint,10pt]{elsarticle}
\usepackage{lipsum}
\makeatletter
\def\ps@pprintTitle{%
  \let\@oddhead\@empty
  \let\@evenhead\@empty
  \def\@oddfoot{}%
  \let\@evenfoot\@oddfoot}

\usepackage{amssymb}
\usepackage{amsthm}
\usepackage{lineno}
\usepackage[shortlabels]{enumitem}
\usepackage{amsfonts}
\usepackage{amsmath}
\usepackage{graphicx}
\usepackage{cases}
\usepackage{epstopdf}
\usepackage{caption}
\usepackage{algorithmic}
\usepackage{mathabx}
\biboptions{sort&compress}

\ifpdf
  \DeclareGraphicsExtensions{.eps,.pdf,.png,.jpg}
\else
  \DeclareGraphicsExtensions{.eps}
\fi

\usepackage{amssymb}
\usepackage{mathtools}

\usepackage[mathscr]{eucal}
\usepackage{color}
\usepackage{geometry}
\usepackage{tikz}
\usetikzlibrary{shapes,arrows}
\usetikzlibrary{decorations.pathreplacing,angles,quotes}

\usepackage{todonotes}

\usepackage{dsfont}
\usepackage{cleveref}
\crefname{assumption}{Assumption}{Assumptions}
\Crefname{assumption}{Assumption}{Assumptions}
\usepackage{verbatim}
\usepackage{graphicx}
\usepackage{subfig}
\usepackage{tikz}
\usepackage{pgfplots}

\newtheorem{theorem}{Theorem}[section]
\newtheorem{lemma}[theorem]{Lemma}
\newtheorem{proposition}[theorem]{Proposition}
\newtheorem{corollary}[theorem]{Corollary}
\numberwithin{equation}{section}
\newtheorem{definition}{Definition}[section]
\newtheorem{remark}{Remark}
\newtheorem{assumption}{Assumption}[section]
\newtheoremstyle{examplestyle}
  {}{}%
  {\normalfont}%
  {}%
  {\itshape}%
  {.}%
  {0.5em}%
  {}

\theoremstyle{examplestyle}
\newtheorem{example}{Example}
\begin{document}
\begin{frontmatter}
\title{A Pricing and Hedging Framework for Infinite-Dimensional Financial Markets}
\author{Hanchao Liu} 
\address{Department of Decision Sciences, HEC Montr\'eal, Montr\'eal, QC, Canada\\
(email: hanchao.liu@hec.ca)}
\begin{abstract}
In this paper, we provide a comprehensive study of how the classical
Black--Scholes-type market framework behaves in an infinite-dimensional
setting. The asset prices are modeled by stochastic evolution equations
in a separable Hilbert space. The evolution of the risk-free asset and
the corresponding discounting operators are governed by a strongly
continuous group of operators. Accordingly, we consider contingent
claims with Hilbert-space-valued payoffs. The pricing and hedging
of such claims are investigated by deriving the corresponding
Hilbert-space-valued Black--Scholes equations and establishing exact
replication results for both strong and mild solutions. The analysis
also covers the special case in which real-valued contingent claims
are allowed. Further attention is devoted to no-arbitrage and market
completeness. As a particular instance of the general framework,
we study geometric Brownian motions in Hilbert spaces.
These results extend key aspects of the classical Black--Scholes
framework to markets with infinite-dimensional asset-price dynamics.
As a concrete application, the model can be interpreted as a
segmented market.
\end{abstract}
\end{frontmatter}

\section{Introduction}

Infinite-dimensional stochastic models arise naturally in various areas of
mathematical finance, particularly when the market involves an entire
price curve or an infinite collection of asset prices. Examples include
interest-rate term structures and commodity forward-price curves.
Such curves and collections of asset prices can be represented as elements
of separable Hilbert spaces, with their evolution described by
infinite-dimensional stochastic evolution equations. This provides a
mathematical framework for investigating derivative pricing, hedging,
and other fundamental questions in financial markets.

A substantial part of the literature on infinite-dimensional financial markets has developed along two closely related directions. The first concerns interest-rate and term-structure models, particularly those based on the Heath--Jarrow--Morton (HJM) framework, which describes the joint evolution of forward rates across a continuum of maturities. This approach has been further developed using stochastic evolution equations and infinite-dimensional stochastic analysis; see, among others, \cite{heath1992bond,filipovic2001consistency,carmona2006interest,filipovic2010term,taflin2005bond,aihara2005stochastic}.

The second direction focuses on commodity and energy markets, where forward prices are indexed by the delivery time (or delivery period) of the underlying commodity. The resulting price curves can be represented as Hilbert-space-valued stochastic processes, providing a mathematical setting for derivative pricing and hedging; see \cite{benth2014representation,benth2015derivatives,benth2023stochastic}.

Infinite-dimensional pricing equations may also arise in financial markets with finitely many traded assets when the underlying price dynamics exhibit additional structures, such as delay or hereditary dependence. In such models, the infinite-dimensional state space is introduced to incorporate the historical evolution of asset prices, even though the traded asset prices themselves remain finite-dimensional. In particular, European option pricing under hereditary price
structures was studied in \cite{chang1999european,chang2007infinite},
and the corresponding infinite-dimensional Black--Scholes equation
was derived. Delayed stock price models have also been studied in
\cite{arriojas2007delayed}, where a delayed Black--Scholes pricing formula was derived within an infinite-dimensional state-space framework. Infinite-dimensional stochastic models also find applications in areas
of mathematical finance beyond derivative pricing, such as the modeling of limit order book dynamics; see, for example, \cite{cont2021}.

Unlike existing infinite-dimensional models of forward curves, where the Hilbert-space-valued process represents a family of real-valued forward prices indexed by delivery dates or times to maturity, our framework directly models the price of a tradable asset as a Hilbert-space-valued stochastic process. In particular, while the no-arbitrage conditions in forward curve models are typically imposed on individual forward price processes associated with fixed delivery dates, our formulation requires the $H$-valued asset-price process, discounted by the associated group of operators, to be a martingale under a risk-neutral measure. This distinction leads to a different financial interpretation and mathematical framework, in which derivative pricing, self-financing trading strategies, and replication are formulated directly for Hilbert-space-valued asset prices.

To motivate our infinite-dimensional formulation, consider the classical
one-dimensional Black--Scholes model under a risk-neutral measure:
\begin{equation}
dX_t=rX_t\,dt+\sigma X_t\,dW_t,
\end{equation}
where $r$ is the risk-free interest rate and $\sigma>0$ is the volatility.
Its solution can be expressed in the mild form
\begin{equation}
X_t=e^{rt}X_0+\int_0^t e^{r(t-s)}\sigma X_s\,dW_s.
\end{equation}
By trivially viewing $\mathbb R$ as a separable Hilbert space $H=\mathbb R$, the classical Black--Scholes model can be formulated as a stochastic evolution equation admitting a mild solution. In particular, the family ${e^{rt}}_{t\in\mathbb R}$ forms a strongly continuous group on $H$. Its restriction to $t\geq 0$ is the semigroup governing the evolution of the risk-free asset, while its restriction to $t\leq 0$ provides the corresponding discounting operators. Moreover, the Brownian motion $W$ can be viewed as a Wiener process on $H$, and a contingent claim can be represented by a suitable payoff function $H\to H$. This observation motivates extending the classical Black--Scholes framework to a general separable Hilbert space $H$ by replacing $\mathbb R$ with $H$ and the corresponding elements of the classical model with their infinite-dimensional counterparts. In particular, the scalar exponential group $\{e^{rt}\}_{t\in\mathbb R}$ is replaced by a strongly continuous group of bounded linear operators on $H$.

In this paper, we provide a comprehensive study of an abstract extension of the classical Black--Scholes-type market framework to an infinite-dimensional setting. Specifically, we consider an $H$-valued asset-price process governed under a risk-neutral measure by
\begin{equation} \label{eqgene}
dX_t=AX_t\,dt+B(X_t)\,dW_t^{\mathbb Q},
\end{equation}
where $A$ is the generator of a strongly continuous group on $H$ and $W^{\mathbb Q}$ is a $Q$-Wiener process. We study the existence of equivalent martingale measures, option pricing, hedging, and market completeness for Hilbert-space-valued asset-price processes admitting mild or strong solutions. We establish the regularity of the associated value functions, derive the corresponding infinite-dimensional Kolmogorov equations in Hilbert spaces, and obtain the Black--Scholes equations for Hilbert-space-valued contingent claims. Unlike the more commonly studied real-valued Kolmogorov and Black--Scholes equations on infinite-dimensional state spaces, the Kolmogorov and Black--Scholes equations considered in our general framework are Hilbert-space-valued, with both the state space and the range of the solution being infinite-dimensional. We also consider the special case $A=rI$, where the classical Black--Scholes pricing framework extends to real-valued contingent claims on the Hilbert state space.

We then investigate hedging, no-arbitrage, and market completeness.
Exact replication results are established for both strong and mild
solutions, using the corresponding notions of admissible trading
strategies. An appropriate lattice structure is introduced to
formulate arbitrage and examine its relationship with equivalent
risk-neutral measures. Market completeness is characterized in
terms of the driving noise and the volatility operator, and
asymptotic completeness is established through finite-dimensional
noise truncations.

Among the dynamics described by \eqref{eqgene}, a particularly interesting case is what we call geometric Brownian motion in Hilbert space, which provides a natural extension of the classical one-dimensional geometric Brownian motion. This process can be viewed as a single $H$-valued asset-price process, while its diagonal structure also allows it to be interpreted as a collection of infinitely many one-dimensional geometric Brownian motions. In this particular setting, the explicit solution of the stochastic evolution equation allows us to weaken the smoothness assumptions required for the corresponding Black--Scholes equation. To this end, we also establish a rigorous proposition justifying the interchange of Fréchet differentiation and Bochner integration. We also provide tractable examples of option pricing for both $H$-valued and real-valued contingent claims within this framework, as well as a possible heuristic financial interpretation of our model in terms of segmented markets.

The remainder of the paper is organized as follows.
Section~\ref{phmild} develops the general pricing and hedging framework,
including the infinite-dimensional Black--Scholes equations and
exact replication results for both strong and mild solutions.
Section~\ref{fundass} examines no-arbitrage and market completeness.
Finally, Section~\ref{GBMsec} studies geometric Brownian motions in
Hilbert space, presents explicit option-pricing examples, and provides
a possible interpretation in terms of segmented markets.

\section{Pricing and Hedging for Mild and Strong Solutions} \label{phmild}
In this section, we develop a pricing and hedging framework for Hilbert-space-valued asset-price processes admitting mild or strong
solutions. As in the classical Black--Scholes model, the contingent claims under consideration have terminal payoffs that are functions
of the terminal asset price $X_T$. After introducing the market dynamics and examining the existence of an equivalent risk-neutral
measure under suitable conditions, we derive the corresponding infinite-dimensional Black--Scholes equation. The replication of such contingent claims is then investigated. Consider the following stochastic differential equation on a Hilbert space $H$:
\begin{equation}
dX_t = \bigl(AX_t+F(X_t)\bigr)\,dt+B(X_t)\,dW_t,
\quad X_0=x\in H,
\end{equation}
where $A:D(A)\subset H\to H$ is the infinitesimal generator
of a strongly continuous group $(S(t))_{t\in\mathbb R}$
on $H$ (see \cite[Section~1.6]{pazy2012semigroups}),
and $F:H\to H$ is the drift coefficient.
Moreover, $W$ is a $Q$-Wiener process,
where $Q$ is a positive trace-class operator on a separable Hilbert space $U$ with eigenvalues $\{\lambda_i\}_{i\geq1}$ and corresponding orthonormal eigenvectors $\{f_i\}_{i\geq1}$, satisfying
$\lambda_i>0$ for all $i\geq1$, and
$U_0=Q^{1/2}(U)$. Also, $B\in\mathcal{L}(H,L_2(U_0,H))$ is a general linear operator representing the volatility coefficient,
where $L_2(U_0,H)$ denotes the space of Hilbert--Schmidt operators from $U_0$ to $H$. We assume that $(\mathcal F_t)_{t\in[0,T]}$ is the usual
augmentation of the natural filtration generated by $W$. Under standard assumptions, the corresponding mild solution is
\begin{equation}\label{Pmild}
X_t
=
S(t)x
+\int_0^t S(t-r)F(X_r)\,dr
+\int_0^t S(t-r)B(X_r)\,dW_r.
\end{equation}
Equivalently, $X$ is also a weak solution in the sense that, for every
$h\in D(A^*)$,
\begin{equation}\label{Qweak}
\langle X_t,h\rangle_H=\langle x,h\rangle_H+\int_0^t\bigl(\langle X_r,A^*h\rangle_H+\langle F(X_r),h\rangle_H\bigr)\,dr+\int_0^t\langle B(X_r)^*h,dW_r\rangle_{U_0}.
\end{equation}
Let $\{e_i\}_{i\geq 1}$ be an orthonormal basis of $H$. Then $X_t$ admits the expansion $X_t=\sum_{i=1}^{\infty}\langle X_t,e_i\rangle_H e_i$, where the series converges in $H$. Thus, $X$ can be viewed as an infinite collection of real-valued stochastic processes. Furthermore, if $e_i\in D(A^*)$, then substituting $h=e_i$ into \eqref{Qweak} shows that $\langle X_t,e_i\rangle_H$ is a real-valued Itô process, although it may still be driven by infinite-dimensional noise. Therefore, the model can equivalently be viewed as describing infinitely many real-valued asset-price processes. This relationship becomes more explicit in the Hilbert-space-valued geometric Brownian motion model studied in Section~\ref{GBMsec}.

As a special case, if $A\in\mathcal L(H)$, then $S(t)$ is uniformly continuous in the operator norm and the mild solution is also a strong solution. The two formulations coincide.
In particular,
\begin{align} \label{xstrong}
X_t
&=
x+\int_0^t\bigl(AX_r+F(X_r)\bigr)\,dr
+\int_0^t B(X_r)\,dW_r
\nonumber\\
&=
S(t)x
+\int_0^t S(t-r)F(X_r)\,dr
+\int_0^t S(t-r)B(X_r)\,dW_r.
\end{align}
In this infinite-dimensional setting, the group $S(t)$ plays the role
of the accumulation factor, analogous to $e^{rt}$ in the classical
scalar framework. Accordingly, its inverse $S(-t)$ can be interpreted
as the corresponding discount operator, and the discounted risky asset
price process is therefore $S(-t)X_t$.  Suppose that there exists a progressively measurable $U_0$-valued
process $\theta=(\theta_t)_{t\in[0,T]}$ such that
\begin{equation}
F(X_t)=B(X_t)\theta_t,
\qquad \text{a.e. } t\in[0,T],\quad \mathbb{P}\text{-a.s.}
\end{equation}
Assume further that $\theta$ satisfies the conditions required for
Girsanov's theorem; see, e.g., \cite[Theorem~10.14 and
Proposition~10.17]{da2014stochastic}. Define the probability measure
$\mathbb{Q}$ by
\begin{equation}
\frac{d\mathbb{Q}}{d\mathbb{P}}
=
\exp\left(
-\int_0^T \langle \theta_t,dW_t\rangle_{U_0}
-\frac{1}{2}\int_0^T\|\theta_t\|_{U_0}^2\,dt
\right).
\end{equation}
Then, under $\mathbb{Q}$, the process
\begin{equation}
W_t^{\mathbb{Q}}
=
W_t+\int_0^t\theta_s\,ds
\end{equation}
is a Wiener process, and the dynamics of the risky asset become
\begin{equation}
dX_t
=
AX_t\,dt+B(X_t)\,dW_t^{\mathbb{Q}}.
\end{equation}
Consequently, its mild representation under $\mathbb{Q}$ is
\begin{equation} \label{Qmild}
X_t
=
S(t)x
+
\int_0^t S(t-r)B(X_r)\,dW_r^{\mathbb{Q}}.
\end{equation}
Thus, the discounted asset price process is given by
\begin{equation}
S(-t)X_t
=
x+\int_0^t S(-r)B(X_r)\,dW_r^{\mathbb Q},
\end{equation}
and hence, under the usual square-integrability condition, 
$\bigl(S(-t)X_t\bigr)_{t\in[0,T]}$ is an $H$-valued martingale under
$\mathbb Q$. When $A\in\mathcal L(H)$, the mild solution \eqref{Qmild} is also a strong solution and satisfies
\begin{equation}\label{Qstrong}
X_t
=
x
+\int_0^t AX_r\,dr
+\int_0^t B(X_r)\,dW_r^{\mathbb Q}.
\end{equation}
Due to the infinite-dimensional nature of the model, there generally exists no equivalent probability measure $\mathbb Q$ under which the discounted asset-price process $e^{-rt}X_t$ is a martingale, whether $X_t$ is a mild or strong solution.
\begin{remark} \label{rnm}
The existence of such a process $\theta$ may be difficult to verify in the
general infinite-dimensional setting. For example, when a suitable
(pseudo-)inverse of $B(x)$ is available, a natural candidate is $\theta_t=B(X_t)^{+}F(X_t)$. However, even when the above expression is well defined, its progressive
measurability and the integrability conditions required for the Girsanov
transformation do not follow automatically from the standard assumptions on
$F$ and $B$ and may require additional structural assumptions.
\end{remark}
\subsection{Infinite-Dimensional Black--Scholes Equation}
We next consider the pricing of a contingent claim with maturity $T$
and $H$-valued payoff $G(X_T)$, where $G:H\to H$ is a given payoff
function. Denote by $V_t$ the value of the claim at time $t$.
As in the classical finite-dimensional setting, the risk-neutral
price is obtained by taking the conditional expectation of the
discounted terminal payoff and then applying the accumulation
operator from time $0$ to time $t$. Accordingly,
\begin{equation}
V_t
=
S(t)\mathbb E^{\mathbb Q}\!\left[
S(-T)G(X_T)\mid\mathcal F_t
\right]
=
S(t-T)\mathbb E^{\mathbb Q}\!\left[
G(X_T)\mid\mathcal F_t
\right].
\end{equation}
By the Markov property, the value process admits the representation
$V_t=V(t,X_t)$, where the pricing function
$V:[0,T]\times H\to H$ is given by
\begin{equation}\label{prifo1}
V(t,x)
=
S(t-T)\,
\mathbb E^{\mathbb Q}\!\left[
G(X_T^{t,x})
\right];
\end{equation}
see \cite[Theorem~9.14]{da2014stochastic}.
In particular, $V(T,x)=G(x)$.

Throughout the paper, we further assume that the usual
augmentations of the natural filtrations generated by
$W$ and $W^{\mathbb Q}$ coincide. We also point out that, in
interest rate models where $X$ represents the forward rate curve
under the Musiela parametrization, the discounted process
$e^{-rt}X_t$ is not required to be a martingale, while the pricing
formula is still given by
$e^{-r(T-t)}\mathbb{E}^{\mathbb{Q}}[G(X_T^{t,x})]$. However, such a pricing formula is not consistent with the current
market framework, since the discounted underlying asset price process
$e^{-rt}X_t$ is not a martingale under $\mathbb Q$.

To derive the corresponding infinite-dimensional Black--Scholes
equation, we first introduce the following notation.
For $T\in L_2(U_0,H)$ and
$S\in\mathcal L(H,\mathcal L(H))$, define
\begin{equation}
\operatorname{Tr}_T(S)
:=
\sum_{k=1}^{\infty}
(S(Te_k))(Te_k),
\end{equation}
where $(e_k)_{k\geq1}$ is an arbitrary orthonormal basis of $U_0$.
The series converges absolutely and is independent of the choice
of orthonormal basis. The pricing equation can be derived informally as follows.\footnote{For
the purpose of this informal derivation, we assume that $X_t$ is
a strong solution of the state equation and that $V(t,x)\in D(A)$.}
Let the bank account be given by $B_t=S(t)e$, where $e\in D(A)$,
and suppose that the value process can be represented as
\begin{equation}
V_t=\phi_t X_t+\Xi_tB_t,
\end{equation}
where $\phi_t$ and $\Xi_t$ are operator-valued processes.
Assume that the strategy is self-financing, namely,
\begin{equation}
dV_t=\phi_t\,dX_t+\Xi_t\,dB_t.
\end{equation}
Under the risk-neutral measure $\mathbb Q$, the risky asset and
the bank account satisfy
\begin{equation}
dX_t=AX_t\,dt+B(X_t)\,dW_t^{\mathbb Q},
\qquad
dB_t=AB_t\,dt.
\end{equation}

Suppose further that $\phi_t$ and $\Xi_t$ commute with $A$
(this is a sufficient condition for \eqref{marcondi},
but is not required for the formal derivation or for
the strategy to be admissible). Then
\begin{equation}
\phi_tAX_t+\Xi_tAB_t
=
A\bigl(\phi_tX_t+\Xi_tB_t\bigr)
=
AV_t,
\end{equation}
and hence the self-financing condition gives
\begin{equation}
dV_t=AV_t\,dt+\phi_tB(X_t)\,dW_t^{\mathbb Q}.
\end{equation}

On the other hand, applying It\^o's formula to $V(t,X_t)$ gives
\begin{align}
dV(t,X_t)
={}&
\Bigg(
\partial_tV(t,X_t)
+
D_xV(t,X_t)AX_t
+
\frac12
\operatorname{Tr}_{B(X_t)}
\left(
D_x^2V(t,X_t)
\right)
\Bigg)dt
\nonumber\\
&+
D_xV(t,X_t)B(X_t)\,dW_t^{\mathbb Q}.
\end{align}
Here, the trace term denotes the contraction of the second
Fr\'echet derivative with the covariance operator induced by
the volatility coefficient. Comparing the drift terms leads to
the corresponding infinite-dimensional Black--Scholes equation
\begin{equation}\label{pde1}
\partial_tV(t,x)
+
D_xV(t,x)[Ax]
+
\frac12
\operatorname{Tr}_{B(x)}
\left(
D_x^2V(t,x)
\right)
-
AV(t,x)
=0,
\end{equation}
with terminal condition $V(T,x)=G(x)$. The remainder of this subsection is devoted to establishing
rigorously that, under appropriate assumptions, the pricing
function given by \eqref{prifo1} satisfies \eqref{pde1}.
\begin{assumption}\label{ass:pricing_regular}
The following regularity conditions hold:
\begin{enumerate}
    \item For every $(t,x)\in[0,T]\times H$,
    $\mathbb E^{\mathbb Q}\!\left[G(X_T^{t,x})\right]\in D(A)$\footnote{For instance, this condition is satisfied if the range of $G$
    is contained in a closed linear subspace $E$ of $H$ such that
    $E\subset D(A)$.}

    \item The payoff function $G:H\to H$ is twice Fr\'echet differentiable,
    and $G$, $D_xG$, and $D_{xx}^2G$ are bounded and Lipschitz continuous.

\item For every $n\ge1$, $B\in\operatorname{Hom}(A_n,\widetilde A_n)$, where $\widetilde A_n$ denotes the operator induced by $A_n$ on
$L_2(U_0,H)$.
    \item For every $x\in H$ and every $e_k$ in an orthonormal basis of
    $U_0$, $X_t^{B(x)e_k}=B(X_t^x)e_k,
    \quad t\in[0,T]$.
 
\end{enumerate}
\end{assumption}
We first present the following standard lemma without proof.
\begin{lemma}\label[lemma]{comlem}
Let $K\subset\mathbb R$ be compact, and let
$T_n,T:K\to L(E)$ be families of bounded linear operators on a Banach space
$E$ such that
\begin{equation}
\sup_{n\ge1}\sup_{t\in K}\|T_n(t)\|_{L(E)}<\infty,
\qquad
\sup_{t\in K}\|T(t)\|_{L(E)}<\infty,
\end{equation}
and, for every $x\in E$,
\begin{equation}
\sup_{t\in K}
\|T_n(t)x-T(t)x\|_E
\longrightarrow0.
\end{equation}
Then, for every compact subset $C\subset E$,
\begin{equation}
\sup_{t\in K}\sup_{x\in C}
\|T_n(t)x-T(t)x\|_E
\longrightarrow0.
\end{equation}
\end{lemma}
Let $(A_n)_{n\ge1}$ be the Yosida approximations of $A$, and let
$(S_n(t))_{t\ge0}$ be the corresponding semigroups. Then
$S_n(t)\to S(t)$ strongly, uniformly on $[0,T]$. For every $x\in H$, define $y_t^x:=S(-t)X_t^x$. Then $y^x$ satisfies
\begin{equation} \label{dyy}
y_t^x
=
x+\int_0^t S(-s)B(X_s^x)\,dW_s^{\mathbb Q}
=
x+\int_0^t S(-s)B(S(s)y_s^x)\,dW_s^{\mathbb Q}.
\end{equation}
Since $t\mapsto y_t^x=S(-t)X_t^x$ is continuous from $[0,T]$ into
$L^2(\Omega;H)$, its range is compact in $L^2(\Omega;H)$.
Hence, by \Cref{comlem},
\begin{equation} \label{convsn}
\sup_{t\in[0,T]}
\mathbb E^{\mathbb Q}
\left[
\left\|
\bigl(S_n(t)-S(t)\bigr)S(-t)X_t^x
\right\|_H^2
\right]
\longrightarrow 0.
\end{equation}
The following lemma allows us to apply \eqref{convsn} to
$L^2(\Omega;L_2(U_0,H))$-valued random variables.
\begin{lemma}\label[lemma]{conversnl2}
Let $K\subset[0,\infty)$ be compact. Then, for every
$R\in L^2(\Omega;L_2(U_0,H))$,
\begin{equation}
\sup_{t\in K}
\mathbb E
\left[
\|(S_n(t)-S(t))R\|_{L_2(U_0,H)}^2
\right]
\longrightarrow0.
\end{equation}
\end{lemma}

\begin{proof}
Let $(e_k)_{k\ge1}$ be an orthonormal basis of $U_0$. For almost every
$\omega$,
\begin{equation}
\sup_{t\in K}
\left\lVert(S_n(t)-S(t))R(\omega)\right\rVert_{L_2(U_0,H)}^2
\le
\sum_{k=1}^\infty
\sup_{t\in K}
\left\lVert(S_n(t)-S(t))R(\omega)e_k\right\rVert_H^2.
\end{equation}
For every $k$, the summand converges to zero, while it is bounded by $C_K\|R(\omega)e_k\|_H^2$, with $C_K$ independent of $n$. Hence, by dominated convergence,
\begin{equation}
\sup_{t\in K}
\|(S_n(t)-S(t))R(\omega)\|_{L_2(U_0,H)}^2
\longrightarrow 0,\quad \text{a.s.}
\end{equation}
Moreover,
\begin{equation}
\sup_{t\in K}
\|(S_n(t)-S(t))R(\omega)\|_{L_2(U_0,H)}^2
\le
C_K\|R(\omega)\|_{L_2(U_0,H)}^2.
\end{equation}
Since $R\in L^2(\Omega;L_2(U_0,H))$, another application of dominated
convergence gives
\begin{equation}
\mathbb E
\left[
\sup_{t\in K}
\|(S_n(t)-S(t))R\|_{L_2(U_0,H)}^2
\right]
\longrightarrow0.
\end{equation}
Therefore,
\begin{equation}
\sup_{t\in K}
\mathbb E
\left[
\|(S_n(t)-S(t))R\|_{L_2(U_0,H)}^2
\right]
\longrightarrow0.
\end{equation}
\end{proof}
\begin{lemma}\label{frecmi}
Suppose that $X_t$ is given by \eqref{Qmild}, and let
\begin{equation}
u(t,x):=\mathbb E^{\mathbb Q}\left[G(X_t^x)\right].
\end{equation}
Assume that $G'$ and $G''$ are Lipschitz continuous. Then, for each
$t\ge0$, the mapping $u(t,\cdot):H\to H$ is twice Fr\'echet
differentiable, and
\begin{equation}
D_xu(t,x)[h]
=
\mathbb E^{\mathbb Q}
\left[
G'(X_t^x)X_t^h
\right],
\qquad h\in H,
\end{equation}
and
\begin{equation}
D_x^2u(t,x)[h]k
=
\mathbb E^{\mathbb Q}
\left[
\bigl(G''(X_t^x)X_t^h\bigr)X_t^k
\right],
\qquad h,k\in H.
\end{equation}
\end{lemma}
\begin{proof}
Since the equation is linear in the initial condition, $X_t^{x+h}=X_t^x+X_t^h$. Moreover, for every $p\ge2$,
\begin{equation}
\mathbb E^{\mathbb Q}
\left[
\lVert X_t^h\rVert_H^p
\right]
\le
C_{p,t}\lVert h\rVert_H^p.
\end{equation}

Let $L_1$ denote the Lipschitz constant of $G'$. By the fundamental
theorem of calculus,
\begin{equation}
G(X_t^x+X_t^h)-G(X_t^x)-G'(X_t^x)X_t^h
=
\int_0^1
\left(
G'(X_t^x+\theta X_t^h)-G'(X_t^x)
\right)X_t^h\,d\theta,
\end{equation}
and hence
\begin{equation}
\left\lVert
G(X_t^x+X_t^h)-G(X_t^x)-G'(X_t^x)X_t^h
\right\rVert_H
\le
\frac{L_1}{2}\left\lVert X_t^h\right\rVert_H^2.
\end{equation}
Therefore,
\begin{equation}
\left\lVert
u(t,x+h)-u(t,x)
-
\mathbb E^{\mathbb Q}
\left[
G'(X_t^x)X_t^h
\right]
\right\rVert_H
\le
\frac{L_1}{2}
\mathbb E^{\mathbb Q}
\left[
\left\lVert X_t^h\right\rVert_H^2
\right]
\le
C_t\left\lVert h\right\rVert_H^2.
\end{equation}
Dividing by $\lVert h\rVert_H$ and letting $h\to0$ gives
\begin{equation}
D_xu(t,x)[h]
=
\mathbb E^{\mathbb Q}
\left[
G'(X_t^x)X_t^h
\right].
\end{equation}

We next differentiate $D_xu(t,\cdot)$. For $h,k\in H$,
\begin{equation}
D_xu(t,x+h)[k]-D_xu(t,x)[k]
=
\mathbb E^{\mathbb Q}
\left[
\left(
G'(X_t^x+X_t^h)-G'(X_t^x)
\right)X_t^k
\right].
\end{equation}
Let $L_2$ denote the Lipschitz constant of $G''$. Again by the
fundamental theorem of calculus,
\begin{equation}
G'(X_t^x+X_t^h)-G'(X_t^x)-G''(X_t^x)X_t^h
=
\int_0^1
\left(
G''(X_t^x+\theta X_t^h)-G''(X_t^x)
\right)X_t^h\,d\theta,
\end{equation}
so that
\begin{equation}
\left\lVert
G'(X_t^x+X_t^h)-G'(X_t^x)-G''(X_t^x)X_t^h
\right\rVert_{\mathcal L(H)}
\le
\frac{L_2}{2}\left\lVert X_t^h\right\rVert_H^2.
\end{equation}
Consequently,
\begin{align}
&
\left\lVert
D_xu(t,x+h)-D_xu(t,x)
-
\mathbb E^{\mathbb Q}
\left[
G''(X_t^x)X_t^h X_t^{(\cdot)}
\right]
\right\rVert_{\mathcal L(H)}
\nonumber\\
\le{}&
\frac{L_2}{2}
\sup_{\left\lVert k\right\rVert_H\le1}
\mathbb E^{\mathbb Q}
\left[
\left\lVert X_t^h\right\rVert_H^2
\left\lVert X_t^k\right\rVert_H
\right]
\nonumber\\
\le{}&
\frac{L_2}{2}
\left(
\mathbb E^{\mathbb Q}
\left[
\left\lVert X_t^h\right\rVert_H^4
\right]
\right)^{1/2}
\sup_{\left\lVert k\right\rVert_H\le1}
\left(
\mathbb E^{\mathbb Q}
\left[
\left\lVert X_t^k\right\rVert_H^2
\right]
\right)^{1/2}
\nonumber\\
\le{}&
C_t\left\lVert h\right\rVert_H^2.
\end{align}
Dividing by $\left\lVert h\right\rVert_H$ and letting $h\to0$ proves that
$D_xu(t,\cdot)$ is Fr\'echet differentiable and
\begin{equation}
D_x^2u(t,x)[h]k
=
\mathbb E^{\mathbb Q}
\left[
\bigl(G''(X_t^x)X_t^h\bigr)X_t^k
\right].
\end{equation}
\end{proof}
The above lemma also remains valid for any mapping $G:H\to E$, where $E$ is a separable Banach space, under the same smoothness assumptions on $G$. We next establish the Kolmogorov equation associated with the value function. The following result provides the required regularity and characterizes the value function as a classical solution. Our proof takes advantage of the fact that $S(t)$ forms a strongly continuous group. We expect the result to remain valid for general strongly continuous semigroups and that Assumption \ref{ass:pricing_regular} (3)--(4) may also be weakened or modified, although the proof would require more involved arguments. For real-valued Kolmogorov equations in Hilbert spaces, see \cite{da2007kolmogorov}.
\begin{theorem}\label{lem:kolmogorov}
Under Assumption~\ref{ass:pricing_regular}, let
\begin{equation}
u(t,x):=\mathbb E^{\mathbb Q}\!\left[G(X_t^x)\right].
\end{equation}
Then, for $(t,x)\in(0,T)\times D(A)$,
\begin{equation}\label{kolmogorov1}
\partial_tu(t,x)
=
D_xu(t,x)[Ax]
+
\frac{1}{2}
\operatorname{Tr}_{B(x)}
\left(D_x^2u(t,x)\right),
\end{equation}
with initial condition $u(0,x)=G(x)$.
\end{theorem}

\begin{proof}
For the time derivative, set
\begin{equation}
\Gamma_t^n:=S_n(t)S(-t)B(X_t^x).
\end{equation}
By \eqref{dyy}, $y^x$ is an $H$-valued diffusion, and since
$A_n\in L(H)$, the process $S_n(t)y_t^x$ satisfies
\begin{equation}
d\bigl(S_n(t)y_t^x\bigr)
=
A_nS_n(t)y_t^x\,dt
+
S_n(t)S(-t)B(X_t^x)\,dW_t^{\mathbb Q}
=
A_nS_n(t)y_t^x\,dt+\Gamma_t^n\,dW_t^{\mathbb Q}.
\end{equation}
Therefore, applying It\^{o}'s formula to $G(S_n(s)y_s^x)$ on
$[t,t+h]$ and taking expectations, the stochastic integral vanishes,
and hence
\begin{align}
u_n(t+h,x)-u_n(t,x)
={}&
\int_t^{t+h}
\mathbb E^{\mathbb Q}
\left[
G'(S_n(s)y_s^x)A_nS_n(s)y_s^x
\right]\,ds
\nonumber\\
&+
\frac12
\int_t^{t+h}
\mathbb E^{\mathbb Q}
\left[
\operatorname{Tr}_{\Gamma_s^n}
\left(G''(S_n(s)y_s^x)\right)
\right]\,ds,
\end{align}
where $u_n(t,x):=
\mathbb E^{\mathbb Q}
\left[
G(S_n(t)y_t^x)
\right]$. Therefore,
\begin{equation}\label{pden}
\partial_tu_n(t,x)
=
\mathbb E^{\mathbb Q}
\left[
G'(S_n(t)y_t^x)A_nS_n(t)y_t^x
\right]
+
\frac12
\mathbb E^{\mathbb Q}
\left[
\operatorname{Tr}_{\Gamma_t^n}
\left(G''(S_n(t)y_t^x)\right)
\right].
\end{equation}

By Lemma~\ref{frecmi}, we have
\begin{equation}
D_xu(t,x)[h]
=
\mathbb E^{\mathbb Q}\!\left[
G'(X_t^x)X_t^h
\right],
\qquad h\in H,
\end{equation}
and
\begin{equation}
\bigl(D_x^2u(t,x)[h]\bigr)k
=
\mathbb E^{\mathbb Q}\!\left[
\bigl(G''(X_t^x)[X_t^h]\bigr)X_t^k
\right],
\qquad h,k\in H.
\end{equation}

Let $x\in D(A)$. Since $A_nx\to Ax$ in $H$, we write
\begin{equation}
S_n(t)y_t^{A_nx}-S(t)y_t^{Ax}
=
S_n(t)\left(y_t^{A_nx}-y_t^{Ax}\right)
+
\left(S_n(t)-S(t)\right)y_t^{Ax}.
\end{equation}
Hence,
\begin{align}
\sup_{t\in[0,T]}
\mathbb E^{\mathbb Q}
\left[
\left\lVert
S_n(t)y_t^{A_nx}-S(t)y_t^{Ax}
\right\rVert_H^2
\right]
\le
&2\sup_{t\in[0,T]}
\mathbb E^{\mathbb Q}
\left[
\left\lVert
S_n(t)\left(y_t^{A_nx}-y_t^{Ax}\right)
\right\rVert_H^2
\right]
\notag\\ &+
2\sup_{t\in[0,T]}
\mathbb E^{\mathbb Q}
\left[
\left\lVert
\left(S_n(t)-S(t)\right)y_t^{Ax}
\right\rVert_H^2
\right].
\end{align}
The first term converges to zero by the standard uniform $L^2$-estimate,
$A_nx\to Ax$, and the uniform boundedness of $S_n(t)$, while
the second term converges to zero by \eqref{convsn}, applied to the fixed
initial condition $Ax$. Therefore,
\begin{equation}
\sup_{t\in[0,T]}
\mathbb E^{\mathbb Q}
\left[
\left\|
S_n(t)y_t^{A_nx}-S(t)y_t^{Ax}
\right\|_H^2
\right]
\longrightarrow0.
\end{equation}
Moreover, since $x\in D(A)$, the sequence $\{A_nx\}_{n\ge1}$ is bounded
in $H$, and hence
\begin{equation}
\sup_{n\ge1}\sup_{t\in[0,T]}
\mathbb E^{\mathbb Q}
\left[
\left\|
S_n(t)y_t^{A_nx}
\right\|_H^2
\right]
<\infty.
\end{equation}
By Assumption~\ref{ass:pricing_regular}(3),
$A_nS_n(t)y_t^x=S_n(t)y_t^{A_nx}$, and hence
\begin{align}
&
\sup_{t\in[0,T]}
\left\|
\mathbb E^{\mathbb Q}
\left[
G'(S_n(t)y_t^x)S_n(t)y_t^{A_nx}
\right]
-
\mathbb E^{\mathbb Q}
\left[
G'(S(t)y_t^x)S(t)y_t^{Ax}
\right]
\right\|_H
\nonumber\\
\le{}&
L_{G'}
\left(
\sup_{t\in[0,T]}
\mathbb E^{\mathbb Q}
\|S_n(t)y_t^x-S(t)y_t^x\|_H^2
\right)^{1/2}
\left(
\sup_{t\in[0,T]}
\mathbb E^{\mathbb Q}
\|S_n(t)y_t^{A_nx}\|_H^2
\right)^{1/2}
\nonumber\\
&+
\|G'\|_\infty
\left(
\sup_{t\in[0,T]}
\mathbb E^{\mathbb Q}
\|S_n(t)y_t^{A_nx}-S(t)y_t^{Ax}\|_H^2
\right)^{1/2}
\longrightarrow0.
\end{align}
Therefore,
\begin{equation}\label{converfir}
D_xu_n(t,x)[A_nx]
=
\mathbb E^{\mathbb Q}
\left[
G'\left(S_n(t)y_t^x\right)
S_n(t)y_t^{A_nx}
\right]
\longrightarrow
\mathbb E^{\mathbb Q}
\left[
G'\left(S(t)y_t^x\right)
S(t)y_t^{Ax}
\right]
=
D_xu(t,x)[Ax]
\end{equation}
uniformly for $t\in[0,T]$. For the second-order term, since
\begin{equation}
\Gamma_t^n-B(X_t^x)
=
\bigl(S_n(t)-S(t)\bigr)S(-t)B(X_t^x),
\end{equation}
and the mapping
$t\mapsto S(-t)B(X_t^x)$ is continuous from $[0,T]$ into
$L^2(\Omega;L_2(U_0,H))$, its range over $[0,T]$ is compact in
$L^2(\Omega;L_2(U_0,H))$. 
By \eqref{convsn}, \Cref{comlem}, and \Cref{conversnl2},
\begin{equation}
\sup_{t\in[0,T]}
\mathbb E^{\mathbb Q}
\left[
\Vert\Gamma_t^n-B(X_t^x)\Vert_{L_2(U_0,H)}^2
\right]
\longrightarrow0.
\end{equation}
Moreover,
\begin{equation}
\sup_{n\ge1}\sup_{t\in[0,T]}
\mathbb E^{\mathbb Q}
\left[
\|\Gamma_t^n\|_{L_2(U_0,H)}^2
\right]
<\infty.
\end{equation}

Using the trilinearity of
$(T,R,S)\mapsto\operatorname{Tr}_{R,S}(T)$, we have
\begin{align}
&
\operatorname{Tr}_{\Gamma_t^n}
\left(G''(S_n(t)y_t^x)\right)
-
\operatorname{Tr}_{B(X_t^x)}
\left(G''(X_t^x)\right)
\nonumber\\
={}&
\operatorname{Tr}_{\Gamma_t^n,\Gamma_t^n}
\left(G''(S_n(t)y_t^x)-G''(X_t^x)\right)
+
\operatorname{Tr}_{\Gamma_t^n,\Gamma_t^n-B(X_t^x)}
\left(G''(X_t^x)\right)
\nonumber\\
&+
\operatorname{Tr}_{\Gamma_t^n-B(X_t^x),B(X_t^x)}
\left(G''(X_t^x)\right).
\end{align}
Hence,
\begin{align}
&
\sup_{t\in[0,T]}
\mathbb E^{\mathbb Q}\!\left[
\left\|
\operatorname{Tr}_{\Gamma_t^n}
\left(
G''(S_n(t)y_t^x)
\right)
-
\operatorname{Tr}_{B(X_t^x)}
\left(
G''(X_t^x)
\right)
\right\|_H
\right]
\nonumber\\
\le{}&
\sup_{t\in[0,T]}
\mathbb E^{\mathbb Q}\!\left[
\left\|
G''(S_n(t)y_t^x)-G''(X_t^x)
\right\|
\|\Gamma_t^n\|_{L_2(U_0,H)}^2
\right]
\nonumber\\
&+
\sup_{t\in[0,T]}
\mathbb E^{\mathbb Q}\!\left[
\|G''(X_t^x)\|
\|\Gamma_t^n-B(X_t^x)\|_{L_2(U_0,H)}
\|\Gamma_t^n\|_{L_2(U_0,H)}
\right]
\nonumber\\
&+
\sup_{t\in[0,T]}
\mathbb E^{\mathbb Q}\!\left[
\|G''(X_t^x)\|
\|\Gamma_t^n-B(X_t^x)\|_{L_2(U_0,H)}
\|B(X_t^x)\|_{L_2(U_0,H)}
\right].
\end{align}

Since $G''$ is Lipschitz and
\begin{equation}
\|\Gamma_t^n\|_{L_2(U_0,H)}
\le C_T\|B(X_t^x)\|_{L_2(U_0,H)}
\le C_T\|y_t^x\|_H,
\end{equation}
the first term is bounded by
\begin{equation}
C_T
\left(
\sup_{t\in[0,T]}
\mathbb E^{\mathbb Q}
\left[
\|S_n(t)y_t^x-S(t)y_t^x\|_H^2
\right]
\right)^{1/2}
\left(
\sup_{t\in[0,T]}
\mathbb E^{\mathbb Q}
\left[
\|y_t^x\|_H^4
\right]
\right)^{1/2},
\end{equation}
which converges to zero by \eqref{convsn} and the fourth-moment estimate. For the remaining two terms, the boundedness of $G''$ and the
Cauchy--Schwarz inequality give
\begin{align}
&
\|G''\|_\infty
\left(
\sup_{t\in[0,T]}
\mathbb E^{\mathbb Q}
\|\Gamma_t^n-B(X_t^x)\|_{L_2(U_0,H)}^2
\right)^{1/2}
\nonumber\\
&\qquad\times
\left[
\left(
\sup_{t\in[0,T]}
\mathbb E^{\mathbb Q}
\|\Gamma_t^n\|_{L_2(U_0,H)}^2
\right)^{1/2}
+
\left(
\sup_{t\in[0,T]}
\mathbb E^{\mathbb Q}
\|B(X_t^x)\|_{L_2(U_0,H)}^2
\right)^{1/2}
\right],
\end{align}
which also converges to zero. Consequently,
\begin{equation}
\sup_{t\in[0,T]}
\mathbb E^{\mathbb Q}\!\left[
\left\|
\operatorname{Tr}_{\Gamma_t^n}
\left(
G''(S_n(t)y_t^x)
\right)
-
\operatorname{Tr}_{B(X_t^x)}
\left(
G''(X_t^x)
\right)
\right\|_H
\right]
\longrightarrow0.
\end{equation}
Therefore,
\begin{equation}\label{convsec}
\frac12
\mathbb E^{\mathbb Q}
\left[
\operatorname{Tr}_{\Gamma_t^n}
\left(
G''(S_n(t)y_t^x)
\right)
\right]
\longrightarrow
\frac12
\mathbb E^{\mathbb Q}
\left[
\operatorname{Tr}_{B(X_t^x)}
\left(
G''(X_t^x)
\right)
\right]
=
\frac12
\operatorname{Tr}_{B(x)}
\left(
D_x^2u(t,x)
\right),
\end{equation}
where the last equality follows from
Assumption~\ref{ass:pricing_regular}(4). The convergence is uniform
for $t\in[0,T]$. Moreover, by \eqref{convsn} and the Lipschitz continuity of $G$,
$u_n(\cdot,x)\to u(\cdot,x)$ uniformly on $[0,T]$. Thus, combining
\eqref{converfir}, \eqref{convsec}, and \eqref{pden},
$\partial_tu_n(\cdot,x)$ converges uniformly on compact subsets of
$(0,T)$. Hence, by the standard theorem on uniform convergence of
derivatives, we obtain \eqref{kolmogorov1}.
\end{proof}
As a direct consequence, we obtain the following infinite-dimensional Black--Scholes equation.
\begin{corollary}\label{cor:black_scholes}
Under Assumption~\ref{ass:pricing_regular}, define
\begin{equation}
v(t,x):=
\mathbb E^{\mathbb Q}
\left[
G(X_T^{t,x})
\right],
\qquad (t,x)\in[0,T]\times H,
\end{equation}
and
\begin{equation}
V(t,x):=S(t-T)v(t,x).
\end{equation}
Then, for $(t,x)\in(0,T)\times D(A)$,
\begin{equation}\label{black_scholes}
\partial_tV(t,x)
+
D_xV(t,x)[Ax]
+
\frac12
\operatorname{Tr}_{B(x)}
\left(
D_x^2V(t,x)
\right)
-
AV(t,x)
=0,
\end{equation}
with terminal condition $V(T,x)=G(x)$.
\end{corollary}

\begin{proof}
Since the coefficients of the state equation do not depend explicitly
on time, Corollary~9.16 of \cite{da2014stochastic}\footnote{Although
Corollary~9.16 of \cite{da2014stochastic} is stated for bounded
real-valued functions, the identity extends directly to bounded
$H$-valued functions. Indeed, for an orthonormal basis
$\{e_k\}_{k\ge1}$ of $H$, write
$G(x)=\sum_{k\ge1}\langle G(x),e_k\rangle_H e_k$ and apply the scalar
identity to each bounded real-valued function
$x\mapsto\langle G(x),e_k\rangle_H$. The passage to the series follows
from the Bochner dominated convergence theorem.} gives
\begin{equation}
v(t,x)
=
\mathbb E^{\mathbb Q}
\left[
G(X_T^{t,x})
\right]
=
\mathbb E^{\mathbb Q}
\left[
G(S_{T-t}^{x})
\right]
=
u(T-t,x).
\end{equation}
Therefore, \Cref{lem:kolmogorov} gives
\begin{align}
\partial_tv(t,x)
&=
-\partial_tu(T-t,x)
\nonumber\\
&=
-
D_xu(T-t,x)[Ax]
-
\frac12
\operatorname{Tr}_{B(x)}
\left(
D_x^2u(T-t,x)
\right)
\nonumber\\
&=
-
D_xv(t,x)[Ax]
-
\frac12
\operatorname{Tr}_{B(x)}
\left(
D_x^2v(t,x)
\right).
\end{align}
Hence,
\begin{equation}
\partial_tv(t,x)
+
D_xv(t,x)[Ax]
+
\frac12
\operatorname{Tr}_{B(x)}
\left(
D_x^2v(t,x)
\right)
=0.
\end{equation}

Since $V(t,x)=S(t-T)v(t,x)$, Assumption~\ref{ass:pricing_regular}(1)
gives
\begin{align}
\partial_tV(t,x)
={}&
AV(t,x)
+
S(t-T)\partial_tv(t,x)
\nonumber\\
={}&
AV(t,x)
-
D_xV(t,x)[Ax]
-
\frac12
\operatorname{Tr}_{B(x)}
\left(
D_x^2V(t,x)
\right),
\end{align}
which yields \eqref{black_scholes}. Moreover, $V(T,x)=G(x)$.
\end{proof}
We next define admissible trading strategies for both strong and mild solutions of the asset price equation. We first recall the definition of measurability for operator-valued functions. Let $(S,\mathcal A)$ be a measurable space, and let $X$ and $Y$ be Banach spaces. A mapping $T\to\mathcal L(X,Y)$ is said to be strongly operator-measurable if, for every $x\in X$, the $Y$-valued function $T(s)x$ is measurable. Equivalently, $T$ is measurable with respect to $\sigma_{\mathrm{SO}}(X,Y)$, which is generated by open balls in the strong operator topology\cite{hytonen2016analysis}.
\begin{definition}[\textit{Admissible strategy}]
\label{selfstrong}
Let $(\phi_t,\theta_t)\in
\mathcal{L}(H)\times\mathcal{L}(H)$
be progressively measurable operator-valued processes,
and let the portfolio value process be given by
\begin{equation}\label{portde}
V_t=\phi_tX_t+\theta_tB_t.
\end{equation}

We distinguish the following two cases.

\begin{enumerate}
    \item[\textup{(i)}] \textit{Strong solution.}
    Suppose that $X$ satisfies the strong dynamics
    \eqref{xstrong}. The portfolio $(\phi,\theta)$
    is called a (strongly) admissible strategy if
    \begin{equation} \label{sffinan}
    dV_t=\phi_t\,dX_t+\theta_t\,dB_t,
    \end{equation}
    and
    \begin{equation}  \label{marcondi}
    \phi_t AX_t+\theta_t AB_t=AV_t.
    \end{equation}

    \item[\textup{(ii)}] \textit{Mild solution.}
    Suppose that $X$ satisfies the mild dynamics
    \eqref{Pmild}. The portfolio $(\phi,\theta)$
  is called a (mildly) admissible strategy if
    \begin{equation}
    V_t=S(t)V_0
    +\int_0^t S(t-s)\phi_sF(X_s)\,ds
    +\int_0^t S(t-s)\phi_sB(X_s)\,dW_s,
    \qquad t\in[0,T].
    \end{equation}
\end{enumerate}
\end{definition}

When $X$ is a strong solution, it is also a mild solution.
In this case, condition (i) is stronger than condition (ii):
every (strongly) admissible strategy also satisfies the mild
portfolio dynamics in (ii), whereas the converse does not
hold in general. Under a risk-neutral measure $\mathbb Q$,
the discounted portfolio value process
$\{S(-t)V_t\}_{t\in[0,T]}$ is an $H$-valued martingale
for either notion of admissibility, provided the required
integrability conditions hold. In the strong case,
provided that \eqref{sffinan} and the required integrability
conditions hold, this martingale property is equivalent
to \eqref{marcondi}. When $A=rI$ for some $r\in\mathbb R$, condition \eqref{marcondi} holds automatically, and conditions (i) and (ii) are equivalent.

The self-financing condition \eqref{sffinan} relies on the
differential dynamics of the asset price process and is
therefore directly available when $X$ is a strong solution.
For a general mild solution, such a differential representation
need not be available, and the self-financing condition cannot
in general be formulated in this form. For the mild formulation \eqref{Qmild}, the claim $\xi$
is said to be \emph{(mildly) attainable} if there exists
a (mildly) admissible portfolio $(\phi,\theta)$ satisfying
Definition~\ref{selfstrong}(ii). Similarly, for the strong
formulation \eqref{Qstrong}, the claim $\xi$ is said to be
\emph{(strongly) attainable} if there exists a (strongly)
admissible portfolio $(\phi,\theta)$ satisfying
Definition~\ref{selfstrong}(i). In either case, the portfolio
must satisfy
\begin{equation}
\phi_T X_T^x+\theta_T S(T)e=\xi
\qquad\text{a.s.}
\end{equation}
We next investigate the attainability of the contingent claim $G(X_T)$. In view of the two notions of admissible strategies introduced above, we establish two corresponding replication results, one for strong solutions and the other for mild solutions. We first consider the strong solution case. 
\begin{assumption}\label{ABlinear}
For \eqref{Qstrong}, assume that
$Ae\notin\operatorname{span}\{e\}$.
\end{assumption}
Under these assumptions, it is
easy to see that $B_t$ and $AB_t$ are linearly independent for every $t\geq 0$. This condition is not required
in the special case $A=rI$, where the portfolio is defined differently. Let $E\in\{\mathbb R,H\}$. A function $V:[0,T]\times E\to H$ is said to be of class
$C^{1,2}$ if it is differentiable in the first variable and
twice Fréchet differentiable in the second variable, and
$V$, $\partial_t V$, $D_xV$, and $D_x^2V$ are continuous
on $[0,T]\times H$. It is not difficult to verify, using \Cref{frecmi} and \Cref{lem:kolmogorov}, that \Cref{ass:pricing_regular} (2) is sufficient to ensure that $V$ is of class $C^{1,2}$ .
\begin{theorem}\label{thm:replication-strong}
Suppose that the market is described by \eqref{Qstrong},
Assumption~\ref{ABlinear} holds, and $V$ is of class $C^{1,2}$.
Then there exists an admissible strategy $(\phi,\theta)$ satisfying condition (i) of Definition~\ref{selfstrong}. In particular, $G(X_T^x)$ is (strongly) attainable.
\end{theorem}

\begin{proof}
Since $X^x$ is a strong solution,
$dB_t=AB_t\,dt$. Since $S(t)$ is invertible and
$Ae\notin\operatorname{span}\{e\}$, $B_t$ and $AB_t$ are
linearly independent. Hence
$M_t:=\operatorname{span}\{B_t,AB_t\}$ is a two-dimensional
subspace of $H$. Recall that $V$ is the pricing function.
By It\^o's formula,
\begin{align}
dV(t,X_t^x)
={}&
\Bigg[
\partial_tV(t,X_t^x)
+
D_xV(t,X_t^x)[AX_t^x]
+
\frac12
\operatorname{Tr}_{B(X_t^x)}
D_x^2V(t,X_t^x)
\Bigg]dt
\nonumber\\
&+
D_xV(t,X_t^x)
B(X_t^x)\,dW_t^{\mathbb Q}.
\label{itoVstrongmarket}
\end{align}
Set $\phi(t)=D_xV(t,X_t^x)$. For
$(t,\omega)\in[0,T]\times\Omega$, define
\begin{align}
R(t,\omega)
&=
V(t,X_t^x(\omega))
-
D_xV(t,X_t^x(\omega))X_t^x(\omega),
\\
Q(t,\omega)
&=
\partial_tV(t,X_t^x(\omega))
+
\frac12
\operatorname{Tr}_{B(X_t^x(\omega))}
D_x^2V(t,X_t^x(\omega)).
\end{align}
Define
\begin{equation}
\widetilde\theta(t,\omega)
\bigl(\alpha B_t+\beta AB_t\bigr)
=
\frac{\langle \alpha B_t,B_t\rangle_H}
{\lVert B_t\rVert_H^2}R(t,\omega)
+
\frac{\langle \beta AB_t,AB_t\rangle_H}
{\lVert AB_t\rVert_H^2}Q(t,\omega),
\qquad
\alpha,\beta\in\mathbb R.
\label{thetatildestrong}
\end{equation}
Equivalently,
$\widetilde\theta(t,\omega)
(\alpha B_t+\beta AB_t)
=\alpha R(t,\omega)+\beta Q(t,\omega)$.
Since $B_t$ and $AB_t$ are linearly independent, this defines a
unique linear operator
$\widetilde\theta(t,\omega)\in\mathcal L(M_t,H)$, where boundedness
follows from the finite dimensionality of $M_t$. Moreover, for every
$\alpha B_t+\beta AB_t\in M_t$, the measurability of
$\widetilde\theta(t,\omega)(\alpha B_t+\beta AB_t)$ follows
directly from that of $R(t,\omega)$ and $Q(t,\omega)$.

Let $P_{M_t}$ denote the orthogonal projection of $H$ onto $M_t$ and
define $\theta(t,\omega)=\widetilde\theta(t,\omega)P_{M_t}$. Then
\begin{equation}
\theta(t)B_t=R(t),
\qquad
\theta(t)AB_t=Q(t).
\label{thetaidentitiesstrong}
\end{equation}
The first identity gives
\begin{equation}
V(t,X_t^x)
=
\phi(t)X_t^x+\theta(t)B_t.
\end{equation}

On the other hand, by \eqref{itoVstrongmarket} and the definition of
$\phi$,
\begin{equation}
dV(t,X_t^x)
-
\phi(t)\,dX_t^x
=
\left[
\partial_tV(t,X_t^x)
+
\frac12
\operatorname{Tr}_{B(X_t^x)}
D_x^2V(t,X_t^x)
\right]dt
=
Q(t)\,dt.
\end{equation}
Using \eqref{thetaidentitiesstrong} and $dB_t=AB_t\,dt$,
\begin{equation}
Q(t)\,dt
=
\theta(t)AB_t\,dt
=
\theta(t)\,dB_t.
\end{equation}
Consequently,
\begin{equation}
dV(t,X_t^x)
=
\phi(t)\,dX_t^x
+
\theta(t)\,dB_t,
\end{equation}
so $(\phi,\theta)$ is self-financing. Since $V(T,x)=G(x)$,
\begin{equation}
\phi(T)X_T^x
+
\theta(T)B_T
=
V(T,X_T^x)
=
G(X_T^x).
\end{equation}
Moreover, by the definition of $V$ and the Markov property,
\begin{equation}
S(-t)V(t,X_t^x)
=
S(-T)\mathbb E^{\mathbb Q}
[G(X_T^x)\mid\mathcal F_t].
\end{equation}
The right-hand side is an $H$-valued
$\mathbb Q$-martingale. Hence,
$(\phi,\theta)$ is admissible.
Therefore, $G(X_T^x)$ is attainable.
\end{proof}

For the mild solution, we show that the contingent claim $G(X_T)$ can be (mildly) replicated by an admissible strategy $(\phi,\theta)$ satisfying condition (ii) of Definition~\ref{selfstrong}. The proof relies on Malliavin calculus tools in Hilbert spaces; see \cite{carmona2006interest,pronk2014tools}.
\begin{theorem}\label{thm:replication-mild}
Suppose that the market is described by \eqref{Qmild}
and that \Cref{ass:pricing_regular}(2) holds%
\footnote{The second-order differentiability assumptions
on $G$ in \Cref{ass:pricing_regular}(2) are not needed
for this replication result.}.
Then there exists an admissible strategy $(\phi,\theta)$
satisfying condition (ii) of Definition~\ref{selfstrong}
such that
\begin{equation}
\phi(T)X_T^x+\theta(T)B_T=G(X_T^x).
\end{equation}
Hence $G(X_T^x)$ is (mildly) attainable.
\end{theorem}

\begin{proof}
Recall that
\begin{equation}
V_t
=
S(t-T)\mathbb E^{\mathbb Q}
\left[G(X_T^x)\mid\mathcal F_t\right].
\end{equation}
Since $B$ is bounded and linear and $G$ satisfies the stated
regularity conditions, the Malliavin chain rule yields
\begin{equation}
D_sG(X_T^x)
=
DG(X_T^x)D_sX_T^x.
\end{equation}
By the Hilbert-valued Clark--Ocone formula,
\begin{equation}
G(X_T^x)
=
\mathbb E^{\mathbb Q}[G(X_T^x)]
+
\int_0^T Z_s\,dW_s^{\mathbb Q},
\end{equation}
where
$Z_s=
\mathbb E^{\mathbb Q}
\left[D_sG(X_T^x)\mid\mathcal F_s\right]$.
Consequently,
\begin{equation}
V_t
=
S(t)V_0
+
\int_0^t S(t-T)Z_s\,dW_s^{\mathbb Q}.
\label{CO-price}
\end{equation}
Let $\phi_t=D_xV(t,X_t^x)$. We first show that
\begin{equation}
\phi_sB(X_s^x)=S(s-T)Z_s,
\qquad
\label{CO-factorization}
\end{equation}
For $s\in[0,T]$ and $h\in H$, let $X^{s,h}$ denote
the mild solution
\begin{equation}\label{sinev}
X_t^{s,h}
=
S(t-s)h
+
\int_s^t S(t-r)B(X_r^{s,h})\,dW_r^{\mathbb Q},
\qquad t\in[s,T].
\end{equation}
By rewriting the pair of equations in \eqref{sinev} as a single
stochastic evolution equation on $H\times H$ with initial condition
$(X_s^x,B(X_s^x)u)$ at time $s$, and then applying
\Cref{frecmi} and the strong Markov property of the resulting
process, we obtain, for every $u\in U_0$,
\begin{equation}
\phi_sB(X_s^x)u
=
S(s-T)\mathbb E^{\mathbb Q}
\left[
DG(X_T^{s,X_s^x})
X_T^{s,B(X_s^x)u}
\mid\mathcal F_s
\right].
\label{CO-price-derivative}
\end{equation}

On the other hand, the Malliavin derivative equation for
mild solutions of stochastic evolution equations gives,
for almost every $s\in[0,T]$ and every $u\in U_0$,
\begin{equation}
D_sX_t^x(u)
=
S(t-s)B(X_s^x)u
+
\int_s^t S(t-r)B(D_sX_r^x(u))\,dW_r^{\mathbb Q},
\qquad t\in[s,T].
\end{equation}
Here we have used the linearity of $B$, which implies
$DB(x)[h]=B(h)$ for all $x,h\in H$.
Thus, $D_sX_t^x(u)$ and $X_t^{s,B(X_s^x)u}$
satisfy the same mild equation with the same initial
condition. By uniqueness,
\begin{equation}
D_sX_T^x(u)=X_T^{s,B(X_s^x)u}.
\end{equation}
Moreover, the flow property gives
$X_T^{s,X_s^x}=X_T^x$.
Substituting these identities into
\eqref{CO-price-derivative}, we obtain
\begin{align}
\phi_sB(X_s^x)u
&=
S(s-T)\mathbb E^{\mathbb Q}
\left[
DG(X_T^x)D_sX_T^x(u)
\mid\mathcal F_s
\right]\nonumber\\
&=
S(s-T)\mathbb E^{\mathbb Q}
\left[
D_sG(X_T^x)(u)
\mid\mathcal F_s
\right]=
S(s-T)Z_su.
\end{align}
Since this holds for every $u\in U_0$,
\eqref{CO-factorization} follows. By the group property and \eqref{CO-factorization},
for $0\le s\le t\le T$,
\begin{equation}
S(t-T)Z_s
=
S(t-s)S(s-T)Z_s
=
S(t-s)\phi_sB(X_s^x).
\end{equation}
Therefore, \eqref{CO-price} becomes
\begin{equation}
V_t
=
S(t)V_0
+
\int_0^t S(t-s)\phi_sB(X_s^x)
\,dW_s^{\mathbb Q}.
\label{CO-mild-replication}
\end{equation}

Since $B_t=S(t)e$ with $e\neq0$,
the group property ensures that $B_t\neq0$
for every $t\in[0,T]$.
Define $\theta_t\in\mathcal L(H)$ by
\begin{equation}
\theta_t h
=
\frac{\langle h,B_t\rangle_H}{\|B_t\|_H^2}
\left(V_t-\phi_tX_t^x\right),
\qquad h\in H.
\label{CO-theta}
\end{equation}
The progressive measurability of $\phi$ and $\theta$ can be easily verified using \Cref{frecmi}. Moreover, \eqref{CO-mild-replication}
shows that $(\phi,\theta)$ satisfies condition (ii)
of Definition~\ref{selfstrong}. At maturity, since $V(T,x)=G(x)$,
\begin{equation}
\phi_TX_T^x+\theta_TB_T
=
V_T
=
G(X_T^x).
\end{equation}
Therefore, $G(X_T^x)$ is attainable.
\end{proof}
An alternative approach is to first establish \Cref{thm:replication-mild} for a sequence of Yosida approximation markets using It\^o's formula for the corresponding strong solutions
\begin{equation}
dX_t^{n,x}
=
A_nX_t^{n,x}\,dt
+
B(X_t^{n,x})\,dW_t^{\mathbb Q},
\qquad
X_0^{n,x}=x,
\end{equation}
where $A_n$ is the Yosida approximation of $A$. The associated pricing functions are given by
\begin{equation}
V_n(t,x)
=
S_n(t-T)
\mathbb E^{\mathbb Q}
\left[
G(X_{T-t}^{n,x})
\right].
\end{equation}
The replication result can then be obtained by passing to the limit under \Cref{ass:pricing_regular}(2). However, this approach is less straightforward than it may initially appear,
since the approximating groups $S_n(t)$ do not necessarily converge strongly
to $S(t)$ for negative times. We provide a sketch of the argument below. Define the portfolio $(\phi^n,\theta^n)$ by
\begin{equation}
\phi_t^n=D_xV_n(t,X_t^{n,x})
\end{equation}
and
\begin{equation}
\theta_t^nh
=
\frac{\langle h,B_t^n\rangle_H}{\|B_t^n\|_H^2}
\left(V_n(t,X_t^{n,x})-\phi_t^nX_t^{n,x}\right),
\end{equation}
where $B_t^n=S_n(t)e$. Since $A_n$ is bounded, applying It\^o's
formula to $V_n(t,X_t^{n,x})$ yields
\begin{equation}
dV_n(t,X_t^{n,x})
=
A_nV_n(t,X_t^{n,x})\,dt
+
\phi_t^nB(X_t^{n,x})\,dW_t^{\mathbb Q}.
\end{equation}
Consequently,
\begin{equation}
\phi_t^nX_t^{n,x}+\theta_t^nB_t^n
=S_n(t-T)\mathbb E^{\mathbb Q}
\left[G(X_T^{n,x})\mid\mathcal F_t\right]
=S_n(t)V_n(0,x)
+\int_0^t S_n(t-s)\phi_s^n
B(X_s^{n,x})\,dW_s^{\mathbb Q}.
\end{equation}
Thus, the portfolio satisfies condition (ii) of
Definition~\ref{selfstrong} in the approximating market. To pass to the limit, we first apply $S_n(T-t)$ to each term
in the above identities. Defining $u_n(t,x)=\mathbb E^{\mathbb Q}
\left[G(X_{T-t}^{n,x})\right]$, and substituting the definitions of $\phi^n$ and $\theta^n$, we obtain
\begin{equation}
\begin{aligned}
&D_xu_n(t,X_t^{n,x})X_t^{n,x}
+\left(
u_n(t,X_t^{n,x})
-D_xu_n(t,X_t^{n,x})X_t^{n,x}
\right)=\mathbb E^{\mathbb Q}
\left[G(X_T^{n,x})\mid\mathcal F_t\right]\\
&=\mathbb E^{\mathbb Q}\left[G(X_T^{n,x})\right]
+\int_0^t D_xu_n(s,X_s^{n,x})
B(X_s^{n,x})\,dW_s^{\mathbb Q}.
\end{aligned}
\end{equation}
Note that the above identities no longer involve the approximating
groups at negative times. By the standard convergence results
for Yosida approximations of stochastic evolution equations,
the derivative representations, and \Cref{ass:pricing_regular}(2),
passing to the limit in $L^2(\Omega;H)$ in each term of the
above identities yields
\begin{equation}
\begin{aligned}
&D_xu(t,X_t^x)X_t^x
+\left(
u(t,X_t^x)-D_xu(t,X_t^x)X_t^x
\right)=\mathbb E^{\mathbb Q}
\left[G(X_T^x)\mid\mathcal F_t\right]\\
&=\mathbb E^{\mathbb Q}\left[G(X_T^x)\right]
+\int_0^t D_xu(s,X_s^x)
B(X_s^x)\,dW_s^{\mathbb Q},
\end{aligned}
\end{equation}
where $u(t,x)=\mathbb E^{\mathbb Q}
\left[G(X_{T-t}^x)\right]$. Applying $S(t-T)$ to each term and defining $\phi$ and $\theta$ as in the preceding proof, we obtain
\begin{equation}
\phi_tX_t^x+\theta_tB_t
=S(t-T)\mathbb E^{\mathbb Q}
\left[G(X_T^x)\mid\mathcal F_t\right]
=S(t)V(0,x)
+\int_0^t S(t-s)\phi_s
B(X_s^x)\,dW_s^{\mathbb Q}.
\end{equation}
Hence, the portfolio $(\phi,\theta)$ satisfies condition (ii)
of Definition~\ref{selfstrong} and replicates $G(X_T^x)$
at maturity. Note that the resulting strategies $\phi$ and $\theta$ are not necessarily the limits of $\phi^n$ and $\theta^n$, respectively.

\subsection{The Case $S(t)=e^I$ and Real-Valued Payoffs}\label{realsemi}
We now consider the special case in which $S(t)=e^{rt}I$ for some
$r\in\mathbb R$. 
The state process admits the strong solution
\begin{equation}\label{Qreal}
X_t^x
=
x
+
\int_0^t rX_s^x\,ds
+
\int_0^t B(X_s^x)\,dW_s^{\mathbb Q}.
\end{equation}
In this setting, the accumulation operator reduces to
scalar multiplication, which allows us to consider real-valued contingent
claims. In particular, we take the payoff function to be $G:H\to\mathbb R$.
Notice that the scalar exponential group $(e^{rt})_{t\in\mathbb R}$ can
be viewed both as a strongly continuous group on $\mathbb R$ and, through
scalar multiplication, as a strongly continuous group on $H$. Thus,
when the payoff is real-valued, the corresponding discounting is given
by the usual scalar factor $e^{-r(T-t)}$. It is therefore natural to define the time-$t$ price of the option with payoff $G(X_T)$ by
\begin{equation}
V(t,x)
=
e^{-r(T-t)}
\mathbb E^{\mathbb Q}
\left[
G(X_{T-t}^x)
\right],
\qquad (t,x)\in[0,T]\times H.
\end{equation}
Equivalently, defining $u(t,x)=\mathbb E^{\mathbb Q}\left[G(X_t^x)\right]$ we have
\begin{equation}
V(t,x)=e^{-r(T-t)}u(T-t,x).
\end{equation}
\begin{assumption}\label{ass:pricing_regular_real}
The following regularity conditions hold:
\begin{enumerate}
    \item The payoff function $G:H\to\mathbb R$ is twice Fr\'echet
    differentiable, and $G$, $D_xG$, and $D_{xx}^2G$ are bounded and
    Lipschitz continuous.

    \item For every $x\in H$ and every $e_k$ in an orthonormal basis of
    $U_0$,
    \begin{equation}
    X_t^{B(x)e_k}=B(X_t^x)e_k,
    \qquad t\in[0,T].
    \end{equation}
\end{enumerate}
\end{assumption}
\begin{theorem}\label{thm:bX_real}
Suppose that $S(t)=e^{rt}I$ and that
Assumption~\ref{ass:pricing_regular_real} holds. Then the option price
\begin{equation}
V(t,x)
=
e^{-r(T-t)}
\mathbb E^{\mathbb Q}
\left[
G(X_{T-t}^x)
\right],
\qquad (t,x)\in[0,T]\times H,
\end{equation}
is a classical solution of the infinite-dimensional Black--Scholes
equation
\begin{equation}
\partial_t V(t,x)
+
rD_xV(t,x)[x]
+
\frac{1}{2}
\operatorname{Tr}_{B(x)}D_x^2V(t,x)
-rV(t,x)
=
0,
\qquad (t,x)\in[0,T)\times H,
\end{equation}
with terminal condition $V(T,x)=G(x),
 x\in H$.
\end{theorem}
\begin{proof}
Define
\begin{equation}
u(t,x)
=
\mathbb E^{\mathbb Q}
\left[
G(X_t^x)
\right].
\end{equation}
By Lemma~\ref{frecmi} and the discussion following it, $u(t,\cdot)$ is
twice Fr\'echet differentiable and, since $G$ is real-valued,
\begin{equation}
D_xu(t,x)[h]
=
\mathbb E^{\mathbb Q}
\left[
\left\langle G'(X_t^x),X_t^h\right\rangle_H
\right],
\qquad
\bigl(D_x^2u(t,x)[h]\bigr)k
=
\mathbb E^{\mathbb Q}
\left[
\left\langle
G''(X_t^x)X_t^h,X_t^k
\right\rangle_H
\right].
\end{equation}

Since $S(t)=e^{rt}I$, the process $S^x$ is a strong solution of
\begin{equation}
dX_t^x
=
rX_t^x\,dt+B(X_t^x)\,dW_t^{\mathbb Q}.
\end{equation}
Applying It\^o's formula to $G(X_s^x)$ on $[t,t+h]$ and taking
expectations yield
\begin{equation}
u(t+h,x)-u(t,x)
=
r\int_t^{t+h}
\mathbb E^{\mathbb Q}
\left[
\left\langle G'(X_s^x),X_s^x\right\rangle_H
\right]ds+
\frac12
\int_t^{t+h}
\mathbb E^{\mathbb Q}
\left[
\operatorname{Tr}_{B(X_s^x)}
G''(X_s^x)
\right]ds.
\end{equation}
Dividing by $h$ and letting $h\downarrow0$ gives
\begin{equation}
\partial_tu(t,x)
=
r\mathbb E^{\mathbb Q}
\left[
\left\langle G'(X_t^x),X_t^x\right\rangle_H
\right]
+
\frac12
\mathbb E^{\mathbb Q}
\left[
\operatorname{Tr}_{B(X_t^x)}
G''(X_t^x)
\right].
\end{equation}

By the linearity of the solution with respect to its initial condition,
\begin{equation}
D_xu(t,x)[x]
=
\mathbb E^{\mathbb Q}
\left[
\left\langle G'(X_t^x),X_t^x\right\rangle_H
\right].
\end{equation}
Moreover, if $(e_k)_{k\ge1}$ is an orthonormal basis of $U_0$, then
Assumption~\ref{ass:pricing_regular_real} and Lemma~\ref{frecmi} imply
\begin{equation}
\operatorname{Tr}_{B(x)}D_x^2u(t,x)
=
\sum_{k\ge1}
\mathbb E^{\mathbb Q}
\left[
\left\langle
G''(X_t^x)X_t^{B(x)e_k},
X_t^{B(x)e_k}
\right\rangle_H
\right]=
\mathbb E^{\mathbb Q}
\left[
\operatorname{Tr}_{B(X_t^x)}
G''(X_t^x)
\right].
\end{equation}
Consequently,
\begin{equation}
\partial_tu(t,x)
=
rD_xu(t,x)[x]
+
\frac12
\operatorname{Tr}_{B(x)}D_x^2u(t,x).
\end{equation}

Finally, $V(t,x)=e^{-r(T-t)}u(T-t,x)$. Therefore,
\begin{align}
\partial_tV(t,x)
&=
rV(t,x)-e^{-r(T-t)}\partial_tu(T-t,x),
\nonumber\\
D_xV(t,x)
&=
e^{-r(T-t)}D_xu(T-t,x),
\qquad
D_x^2V(t,x)
=
e^{-r(T-t)}D_x^2u(T-t,x).
\end{align}
Substituting the equation for $u$ gives
\begin{equation}
\partial_t V(t,x)
+
rD_xV(t,x)[x]
+
\frac12
\operatorname{Tr}_{B(x)}D_x^2V(t,x)
-rV(t,x)
=
0.
\end{equation}
Since $u(0,x)=G(x)$, we also have $V(T,x)=G(x)$. This completes the proof.
\end{proof}
We now consider the corresponding hedging problem. In the present
real-valued setting, a portfolio is a pair $(\phi,\theta)$, where
$\phi_t\in H$ denotes the position in the risky asset and
$\theta_t\in\mathbb R$ denotes the position in the bank account. The
bank account is given by $\beta_t=e^{rt}$. The value of the portfolio is
\begin{equation}
V_t
=
\left\langle \phi_t,X_t^x\right\rangle_H
+
\theta_t\beta_t.
\end{equation}
The portfolio $(\phi,\theta)$ is said to be admissible if
\begin{equation} \label{sfreal}
dV_t
=
\left\langle \phi_t,dX_t^x\right\rangle_H
+
\theta_t\,d\beta_t.
\end{equation}
This condition coincides with the classical self-financing
condition and is analogous to Definition~\ref{selfstrong}(i),
with the operator-valued risky asset position replaced by
an $H'$-valued process.
In the present setting, it is also equivalent to Definition~\ref{selfstrong}(ii),
for which the portfolio dynamics take the form
\begin{equation}
V_t
=
e^{rt}V_0
+
\int_0^t e^{r(t-s)}
\left\langle \phi_s,B(X_s^x)\,dW_s^{\mathbb Q}\right\rangle_H.
\end{equation}
The discounted value process
$\{e^{-rt}V_t\}_{t\in[0,T]}$ is 
therefore a $\mathbb Q$-martingale. 
\begin{proposition}\label{prop:real_hedging}
Suppose that the assumptions of Theorem~\ref{thm:bX_real} hold, and let
$V$ be the option price given therein. Define
\begin{equation}
\phi_t=D_xV(t,X_t^x),
\qquad
\theta_t
=
\frac{
V(t,X_t^x)-\left\langle D_xV(t,X_t^x),X_t^x\right\rangle_H
}{\beta_t},
\qquad
\beta_t=e^{rt}.
\end{equation}
Then the portfolio $(\phi,\theta)$ is self-financing and replicates the
claim $G(X_T^x)$.
\end{proposition}

\begin{proof}
By It\^o's formula and Theorem~\ref{thm:bX_real},
\begin{equation}
dV(t,X_t^x)
=
rV(t,X_t^x)\,dt
+
\left\langle
D_xV(t,X_t^x),
B(X_t^x)\,dW_t^{\mathbb Q}
\right\rangle_H.
\end{equation}
On the other hand, since
\begin{equation}
dX_t^x=rX_t^x\,dt+B(X_t^x)\,dW_t^{\mathbb Q},
\qquad
d\beta_t=r\beta_t\,dt,
\end{equation}
the definition of $(\phi,\theta)$ yields
\begin{align}
\left\langle\phi_t,dX_t^x\right\rangle_H+\theta_t\,d\beta_t
={}&
r\left\langle D_xV(t,X_t^x),X_t^x\right\rangle_H\,dt
+
\left\langle
D_xV(t,X_t^x),
B(X_t^x)\,dW_t^{\mathbb Q}
\right\rangle_H
\nonumber\\
&+
r\left(
V(t,X_t^x)
-
\left\langle D_xV(t,X_t^x),X_t^x\right\rangle_H
\right)dt
=
dV(t,X_t^x).
\end{align}
Moreover,
\begin{equation}
V(t,X_t^x)
=
\left\langle\phi_t,X_t^x\right\rangle_H+\theta_t\beta_t,
\end{equation}
so the portfolio is admissible. Finally,
$V(T,X_T^x)=G(X_T^x)$, and hence $(\phi,\theta)$ replicates the claim.
\end{proof}
\subsection{Some Weak Versions of the Black–Scholes Equation}

Due to the infinite-dimensional nature of the state space, option prices may fail to possess the Fréchet differentiability required by the classical Black–Scholes equation. Nevertheless, they may satisfy weaker versions of this equation. Such formulations can take different forms depending on the structure of the option price. Here, we briefly consider two examples.

**Hilbert-space-valued options.** Let \((e_i)_{i\ge1}\) be an orthonormal basis of \(H\), and write
\[
V(t,x)=\sum_{i=1}^{\infty}V_i(t,x)e_i,
\qquad
V_i(t,x)=\langle V(t,x),e_i\rangle_H.
\]
Suppose that each \(V_i\) is continuously differentiable in time and twice continuously Fréchet differentiable in space, and that the following series converges in \(H\), locally uniformly in \(t\in(0,T)\) for each fixed \(x\in D(A)\). We may consider the coordinatewise weak Black–Scholes equation
\begin{align} \label{wbsh}
\partial_t V(t,x)
&+\sum_{i=1}^{\infty}D_xV_i(t,x)[Ax]e_i
+\frac12\sum_{i,j=1}^{\infty}\lambda_j
D_{xx}V_i(t,x)[B(x)e_j,B(x)e_j]e_i
-AV(t,x)=0.
\end{align}
Unlike the classical formulation, this equation does not require the existence of the Fréchet derivatives of the \(H\)-valued function \(V\).

**Real-valued options.** For \(V:(0,T)\times H\to\mathbb R\), let
\[
P_Nx=\sum_{i=1}^Nx_ie_i,
\qquad
V^N(t,x)=V(t,P_Nx).
\]
For each $N$, suppose that $V^N$ satisfies the
$N$-dimensional Black--Scholes equation
\begin{equation}
\partial_t V^N(t,x)+r\sum_{i=1}^N x_i\partial_i V^N(t,x)+\frac12\sum_{i=1}^N\lambda_i x_i^2\partial_{ii}V^N(t,x)-rV^N(t,x)=0.
\end{equation}
Provided that both $V^N(t,x)\to V(t,x)$ and
$\partial_t V^N(t,x)$ converge locally uniformly
in $t$ for each fixed $x$, passing to the limit
as $N\to\infty$ yields
\begin{equation}\label{wbsr}
\partial_t V(t,x)+\lim_{N\to\infty}\left[r\sum_{i=1}^N x_i\partial_i V^N(t,x)+\frac12\sum_{i=1}^N\lambda_i x_i^2\partial_{ii}V^N(t,x)\right]-rV(t,x)=0.
\end{equation}
\section{No-Arbitrage and Market Completeness}\label{fundass}
As in the standard finite-dimensional market setting, we now turn to the fundamental questions of no-arbitrage and market completeness. In particular, we investigate their relationship with the existence and uniqueness of an equivalent risk-neutral measure in the present infinite-dimensional framework. The preceding hedging results concern contingent claims
of the form $G(X_t)$. In this subsection, we consider general contingent
claims in $L^2(\Omega,\mathcal F_T,\mathbb Q;H)$ and establish the
infinite-dimensional counterparts of the first and second fundamental
theorems of asset pricing. 
\subsection{No-Arbitrage and Risk-Neutral Measures}
Since portfolio values are $H$-valued, we first introduce an order structure on $H$ to capture the notion of a nonnegative payoff and, consequently, of an arbitrage opportunity. We equip $H$ with an order $\preceq$ such that $H$ becomes a Banach lattice
(and hence a Hilbert lattice; see \cite{meyer2012banach}), and denote the
resulting lattice by $(H,\preceq)$. In particular, we write $x\succeq 0$
if $x$ belongs to the positive cone $H_+$, and $x\succ 0$ if
$x\in H_+\setminus\{0\}$.
For the markets determined by \eqref{Qstrong} and \eqref{Qmild},
consider an admissible strategy $(\phi,\theta)$ in the corresponding
framework, according to Definition~\ref{selfstrong}, with value process
$V$. We say that a $\preceq$-arbitrage occurs if
\begin{equation} \label{arbi}
V_0=0,\qquad 0\preceq V_T\quad\text{a.s.},
\qquad \mathbb P(0\prec V_T)>0.
\end{equation}
The market is said to be free of $\preceq$-arbitrage if no such
portfolio exists. In the real-valued setting of Section~\ref{realsemi},
we consider admissible portfolios satisfying \eqref{sfreal}.
The usual order $\leq$ trivially makes $\mathbb R$ a Banach
lattice, and the above definition of arbitrage, with
$\preceq$ replaced by $\leq$, reduces to the classical
notion of arbitrage in finite-dimensional markets.
\begin{proposition}\label{prop:no_arbitrage}
Let $E$ be a separable Banach lattice, let
$\{T(t)\}_{t\in\mathbb R}$ be a $C_0$-group on $E$, and let
$\mathbb Q$ be a probability measure equivalent to $\mathbb P$. Suppose
that $V$ is an $E$-valued process with $V_0=0$ such that
$\{T(-t)V_t\}_{t\in[0,T]}$ is an $E$-valued $\mathbb Q$-martingale.
Then there does not exist $t\in(0,T]$ such that
\begin{equation}
0\preceq V_t\quad\text{a.s.},
\qquad
\mathbb P(0\prec V_t)>0.
\end{equation}
\end{proposition}

\begin{proof}
Suppose, to the contrary, that for some $t\in(0,T]$,
\begin{equation}
0\preceq V_t\quad\text{a.s.},
\qquad
\mathbb P(0\prec V_t)>0.
\end{equation}
Since $\mathbb Q$ is equivalent to $\mathbb P$,
$\mathbb Q(0\prec V_t)>0$. Since $E$ is a separable Banach lattice,
there exists a strictly positive continuous linear functional
$\varphi\in E^*$. Hence,
\begin{equation}
\varphi(V_t)\geq0\quad\text{a.s.},
\qquad
\mathbb Q\bigl(\varphi(V_t)>0\bigr)>0,
\end{equation}
and therefore
\begin{equation}
\varphi\left(\mathbb E^{\mathbb Q}[V_t]\right)
=
\mathbb E^{\mathbb Q}\left[\varphi(V_t)\right]>0.
\end{equation}
Since the positive cone of $E$ is closed and convex,
$\mathbb E^{\mathbb Q}[V_t]$ belongs to the positive cone. Moreover, the
strict positivity of $\varphi$ implies
$\varphi(\mathbb E^{\mathbb Q}[V_t])>0$, and hence
$0\prec\mathbb E^{\mathbb Q}[V_t]$. On the other hand, since $T(-t)V_t$ is a $\mathbb Q$-martingale and
$V_0=0$,
\begin{equation}
T(-t)\mathbb E^{\mathbb Q}[V_t]
=
\mathbb E^{\mathbb Q}[T(-t)V_t]
=0.
\end{equation}
Since $T(-t)$ is invertible, it follows that
$\mathbb E^{\mathbb Q}[V_t]=0$, a contradiction.
\end{proof}
Applying Proposition~\ref{prop:no_arbitrage} to the cases
$E=H$ and $E=\mathbb R$ yields the following corollary.
\begin{corollary}\label{cor:no_arbitrage}
Suppose that there exists a risk-neutral probability measure
$\mathbb Q$ equivalent to $\mathbb P$. Then:
\begin{enumerate}
\item The markets determined by \eqref{Qstrong} and
\eqref{Qmild} are free of $\preceq$-arbitrage.
\item When $A=rI$ for some $r\in\mathbb R$, the market
considered in Section~\ref{realsemi} is free of arbitrage.
\end{enumerate}
\end{corollary}
In general, the converses of the preceding results do not necessarily hold, as they depend heavily on the specific structures of $F$, $B$, and the driving noise $W$, as well as the lattice structure of $H$. Even when $A=rI$, the converse of \Cref{prop:no_arbitrage} is not guaranteed under the present assumptions for either $H$-valued or real-valued portfolios. Here, we use the real-valued case as an illustrative example. Indeed, the absence of arbitrage requires $F(X_t)\in\overline{\operatorname{Ran}B(X_t)}$ for almost every $t\in[0,T]$, almost surely. Otherwise, letting $P_t$ denote the orthogonal projection onto $\ker B(X_t)^*$, we may construct a self-financing portfolio with $\phi(t)=P_tF(X_t)$, which may yield an arbitrage opportunity. However, even if $F(X_t)\in\overline{\operatorname{Ran}B(X_t)}$, this condition does not guarantee the existence of a process $\theta_t$ satisfying $B(X_t)\theta_t=F(X_t)$, since the range of $B(X_t)$ need not be closed in infinite-dimensional spaces. This issue is closely related to the dimension of the driving Wiener process, which we discuss in the next section. A comprehensive characterization of the relationship between the absence of arbitrage and the existence of an equivalent risk-neutral measure would require the development of a Delbaen--Schachermayer-type theorem in the present setting.
\subsection{Market Completeness}
We now consider the analogous notion of market completeness. Our analysis of market completeness is primarily operator-theoretic, focusing on the solvability of the operator equations arising from the martingale representation of contingent claims. As discussed in Remark~\ref{rnm}, the existence of a risk-neutral measure is generally difficult to verify in the present infinite-dimensional setting, and we therefore study completeness under a given risk-neutral measure $\mathbb Q$. When the state process is described by the strong formulation \eqref{Qstrong} or the mild formulation \eqref{Qmild}, the market is called \emph{strongly complete} or \emph{mildly complete}, respectively, if every contingent claim in $L^2(\Omega,\mathcal F_T,\mathbb Q;H)$ is strongly or mildly attainable, respectively, by an admissible portfolio. For the real-valued setting of Section~\ref{realsemi}, the market is
called \emph{complete} if every real-valued contingent claim in
$L^2(\Omega,\mathcal F_T,\mathbb Q)$ is attainable by an admissible portfolio defined in \eqref{sfreal}. Attainability in infinite-dimensional bond markets has also been studied in the literature; see, for instance, \cite{taflin2005bond}. Let $\xi\in L^2(\Omega,\mathcal F_T,\mathbb Q;H)$ be a contingent
claim, and define
\begin{equation}
M_t
=
\mathbb E^{\mathbb Q}
\left[
S(-T)\xi\mid\mathcal F_t
\right].
\end{equation}
By the martingale representation theorem, there exists a
progressively measurable process
$Z:[0,T]\times\Omega\to L_2(U_0,H)$ satisfying
\begin{equation}
\mathbb E^{\mathbb Q}
\int_0^T
\|Z_t\|_{L_2(U_0,H)}^2\,dt<\infty,
\end{equation}
such that
\begin{equation}
M_t=M_0+\int_0^t Z_s\,dW_s^{\mathbb Q}.
\end{equation}

\begin{theorem}\label{thm:general_replication}
Suppose that the market is driven by \eqref{Qmild}, or by
\eqref{Qstrong} with Assumption~\ref{ABlinear} holding.
Then $\xi$ is mildly and strongly attainable, respectively,
if and only if there exists a progressively measurable process
$\phi:[0,T]\times\Omega\to\mathcal L(H)$ satisfying
\begin{equation}
\phi(t)B(X_t)=S(t)Z_t
\end{equation}
for almost every $(t,\omega)$.
\end{theorem}

\begin{proof}
Suppose first that $\xi$ is attainable, and let $(\phi,\theta)$ be an
admissible portfolio such that $V_T=\xi$. Since
$\{S(-t)V_t\}_{t\in[0,T]}$ is an $H$-valued $\mathbb Q$-martingale,
\begin{equation}
S(-t)V_t
=
\mathbb E^{\mathbb Q}
\left[
S(-T)\xi\mid\mathcal F_t
\right]
=
M_t.
\end{equation}
By Definition~\ref{selfstrong}, in both the strong and mild cases, we have
\begin{equation}
S(-t)V_t
=
V_0+\int_0^t S(-s)\phi(s)B(X_s)\,dW_s^{\mathbb Q}.
\end{equation}
Comparing this with
$M_t=M_0+\int_0^t Z_s\,dW_s^{\mathbb Q}$, the uniqueness of the
martingale representation yields
$S(-t)\phi(t)B(X_t)=Z_t$ almost everywhere, and hence
$\phi(t)B(X_t)=S(t)Z_t$.

Conversely, for the strong case, suppose that there exists $\phi(t)\in\mathcal L(H)$ such
that $\phi(t)B(X_t)=S(t)Z_t$, and set $V_t=S(t)M_t$. Then
$V_T=\xi$ and
\begin{equation}
dV_t
=
AV_t\,dt+S(t)Z_t\,dW_t^{\mathbb Q}.
\end{equation}

Choose $\theta(t)\in\mathcal L(H)$, as in
\Cref{thm:replication-strong}, such that
$\theta(t)B_t=V_t-\phi(t)X_t$ and
$\theta(t)AB_t=AV_t-\phi(t)AX_t$. It follows that $V_t=\phi(t)X_t+\theta(t)B_t$, and
\begin{align}
\phi(t)\,dX_t+\theta(t)\,dB_t
&=
\left[
\phi(t)AX_t+\theta(t)AB_t
\right]dt
+
\phi(t)B(X_t)\,dW_t^{\mathbb Q}
\nonumber\\
&=
AV_t\,dt+S(t)Z_t\,dW_t^{\mathbb Q}
=
dV_t.
\end{align}
Moreover, $S(-t)V_t=M_t$, which is an $H$-valued
$\mathbb Q$-martingale. Thus $(\phi,\theta)$ is
admissible, and since $V_T=\xi$, the claim $\xi$ is
attainable.
For the mild case, since $B_t=S(t)e\neq 0$, define
$\theta(t)\in\mathcal L(H)$ by
\begin{equation}
\theta(t)h
=
\frac{\langle h,B_t\rangle_H}{\|B_t\|_H^2}
\bigl(V_t-\phi(t)X_t\bigr),
\qquad h\in H.
\end{equation}
Then $V_t=\phi(t)X_t+\theta(t)B_t$. Moreover,
\begin{equation}
V_t
=S(t)M_t
=S(t)V_0+\int_0^t S(t-s)S(s)Z_s\,dW_s^{\mathbb Q}
=S(t)V_0+\int_0^t S(t-s)\phi(s)B(X_s)\,dW_s^{\mathbb Q}.
\end{equation}
Thus the mild self-financing condition holds. Furthermore,
$S(-t)V_t=M_t$ is an $H$-valued $\mathbb Q$-martingale.
Consequently, $(\phi,\theta)$ is admissible in the mild
framework, and $V_T=\xi$ shows that $\xi$ is mildly attainable.

\end{proof}

Applying the same approach with $A=rI$ and a real-valued
contingent claim $\xi\in L^2(\Omega,\mathcal F_T,\mathbb Q)$,
and denoting by $Z$ the associated $U_0$-valued martingale representation process, we obtain the following result.

\begin{corollary}\label{cor:real_replication}
The claim $\xi\in L^2(\Omega,\mathcal F_T,\mathbb Q)$ is attainable if and only if there exists a
progressively measurable process
$\phi:[0,T]\times\Omega\to H$ satisfying
\begin{equation}
B(X_t^x)^*\phi_t=e^{rt}Z_t
\end{equation}
for almost every $(t,\omega)$.
\end{corollary}
\begin{lemma}\label{lem:measurable_nonrange}
Let $H$ and $U$ be separable Hilbert spaces, with $U$
infinite-dimensional. There exists a bounded Borel measurable mapping
$F:L_2(H,U)\to U$ such that
$F(K)\notin\mathcal R(K)$ for every $K\in L_2(H,U)$.
\end{lemma}

\begin{proof}
Let $\{e_n\}_{n\geq1}$ be an orthonormal basis of $U$. For each
$n\geq1$, define
$f_n(K)=\lVert K^*e_n\rVert_H$, $K\in L_2(H,U)$.
Each $f_n$ is continuous and $f_n(K)\to0$. Set $n_0(K)=0$. For
$m\geq1$, define $n_m(K)$ to be the smallest integer strictly larger
than $n_{m-1}(K)$ such that
$m f_{n_m(K)}(K)<2^{-m-1}$. The preceding convergence shows that
$n_m(K)$ is well defined for every $K$ and every $m$.

We next show inductively that each mapping
$n_m:L_2(H,U)\to\mathbb N$ is Borel measurable. For $m=1$ and
$j\geq1$,
\begin{equation}
\{K:n_1(K)=j\}
=
\{K:f_j(K)<2^{-2}\}
\cap
\bigcap_{k=1}^{j-1}\{K:f_k(K)\geq2^{-2}\},
\end{equation}
which is Borel since each $f_k$ is continuous. Suppose now that
$n_{m-1}$ is Borel measurable. For every $j\geq1$,
\begin{align}
\{K:n_m(K)=j\}
=
\bigcup_{i=0}^{j-1}
\Bigg(
&
\{K:n_{m-1}(K)=i\}
\cap
\{K:m f_j(K)<2^{-m-1}\}
\nonumber\\
&\cap
\bigcap_{k=i+1}^{j-1}
\{K:m f_k(K)\geq2^{-m-1}\}
\Bigg).
\end{align}
Thus every fiber $\{K:n_m(K)=j\}$ is Borel, and hence $n_m$ is Borel
measurable.

Define
$F(K)=\sum_{m=1}^{\infty}2^{-m}e_{n_m(K)}$ and
$F_N(K)=\sum_{m=1}^{N}2^{-m}e_{n_m(K)}$.
Since the mappings $n_m$ are Borel measurable, each $F_N$ is Borel
measurable. Moreover, since $n_1(K)<n_2(K)<\cdots$, the vectors
$\{e_{n_m(K)}\}_{m\geq1}$ are orthonormal, and hence
\begin{equation}
\lVert F(K)-F_N(K)\rVert_U^2
=
\sum_{m=N+1}^{\infty}4^{-m}
\longrightarrow0.
\end{equation}
Therefore $F$ is Borel measurable.

It remains to show that $F(K)\notin\mathcal R(K)$. Suppose, to the
contrary, that $F(K)=Kx$ for some $x\in H$, and choose an integer
$m\geq\lVert x\rVert_H$. By the definition of $F(K)$,
$\langle F(K),e_{n_m(K)}\rangle_U=2^{-m}$, whereas
\begin{align}
\left|
\left\langle Kx,e_{n_m(K)}\right\rangle_U
\right|
&=
\left|
\left\langle x,K^*e_{n_m(K)}\right\rangle_H
\right|
\nonumber\\
&\leq
\lVert x\rVert_H
\lVert K^*e_{n_m(K)}\rVert_H
\leq
m f_{n_m(K)}(K)
<
2^{-m-1},
\end{align}
a contradiction. Therefore $F(K)\notin\mathcal R(K)$.
\end{proof}

\begin{proposition}\label{lem:measurable_factorization}
Let $(S,\mathcal A,\mu)$ be a finite measure space, let $U,H,V$ be separable
Hilbert spaces with $V\neq\{0\}$, and let
$R\in L^2(S;L_2(U,H))$. Then the following statements are equivalent:
\begin{enumerate}
\item For every $Z\in L^2(S;L_2(U,V))$, there exists a
strongly measurable operator-valued function
$T:S\to\mathcal L(H,V)$ such that
\begin{equation}
T(s)R(s)=Z(s)
\end{equation}
for $\mu$-almost every $s\in S$.
\item The space $U$ is finite-dimensional and
$\ker R(s)=\{0\}$ for $\mu$-almost every $s\in S$.
\end{enumerate}
\end{proposition}
\begin{proof}
Suppose that (2) holds, and let
$Z\in L^2(S;L_2(U,V))$. Define
\begin{equation}
G(s)=R(s)^*R(s), \qquad s\in S.
\end{equation}
Since $U$ is finite-dimensional,
$L_2(U,H)=\mathcal L(U,H)$ with equivalent norms, and hence $R$ is
also measurable as an $\mathcal L(U,H)$-valued mapping. Since the
adjoint map is continuous in the operator norm,
$s\mapsto R(s)^*$ is measurable as an $\mathcal L(H,U)$-valued
mapping. Hence $s\mapsto(R(s)^*,R(s))$ is measurable. Since the
composition mapping
\begin{equation}
\mathcal L(H,U)\times\mathcal L(U,H)\ni(K_1,K_2)
\mapsto K_1K_2\in\mathcal L(U)
\end{equation}
is continuous, $G:S\to\mathcal L(U)$ is measurable. Let
\begin{equation}
GL(U)=\{A\in\mathcal L(U):A\text{ is a bijection on }U\} \end{equation}
 Since $U$ is finite-dimensional, $\mathcal R(R(s))$ is closed for
every $s\in S$. Thus, if $R(s)$ is injective, we have
$\mathcal R(R(s)^*)=(\ker R(s))^\perp=U$.
Hence $R(s)^*R(s)$ is a bijection on $U$ \cite{desoer1963note}, and therefore
$(R(s)^*R(s))^{-1}\in\mathcal L(U)$. Conversely, if
$G(s)\in GL(U)$, then $\ker R(s)=\ker G(s)=\{0\}$. Therefore,
\begin{equation} \label{setd1}
D:=\{s\in S:R(s)\text{ is injective}\}
=G^{-1}(GL(U)).
\end{equation}
Since $GL(U)$ is open in $\mathcal L(U)$, we have $D\in\mathcal A$,
and by assumption $\mu(S\setminus D)=0$.

Equipping $D$ with the trace $\sigma$-algebra $\mathcal A_D$, the
restriction $G|_D:(D,\mathcal A_D)\to
(GL(U),\mathcal B(GL(U)))$ is measurable. Since the inversion map from $GL(U)$ to $\mathcal L(U)$ is continuous in the operator norm, $s\mapsto G(s)^{-1}$ is measurable on $D$ as an $\mathcal L(U)$-valued mapping. Moreover, the restrictions of $Z$ and $R^*$ to
$D$ are measurable. Consequently,
$s\mapsto Z(s)G(s)^{-1}R(s)^*$ is an $L_2(H,V)$-valued measurable
mapping on $D$.

Define
\begin{equation} \label{setd}
T(s)=
\begin{cases}
Z(s)(R(s)^*R(s))^{-1}R(s)^*, & s\in D,\\
0, & s\in D^c.
\end{cases}
\end{equation}
Since $D\in\mathcal A$, the mapping
$T:S\to L_2(H,V)$ is measurable and thus also strongly measurable
as an $\mathcal L(H,V)$-valued mapping. For
$s\in D$,
\begin{equation}\label{conTs}
T(s)R(s)
=
Z(s)(R(s)^*R(s))^{-1}R(s)^*R(s)
=
Z(s).
\end{equation}
Since $\mu(S\setminus D)=0$, it follows that
$T(s)R(s)=Z(s)$ for $\mu$-almost every $s\in S$.

Conversely, suppose that (1) holds. We first show that $U$ is
finite-dimensional. Otherwise, by \Cref{lem:measurable_nonrange},
there exists a bounded Borel measurable mapping
$F:L_2(H,U)\to U$ such that
$F(K)\notin\mathcal R(K)$ for every $K\in L_2(H,U)$.
Fix $v\in V$ with $\lVert v\rVert_V=1$ and define
\begin{equation}
Z(s)u
=
\langle u,F(R(s)^*)\rangle_U v.
\end{equation}
Then $Z\in L^2(S;L_2(U,V))$. By (1), there exists a measurable
$T:S\to L_2(H,V)$ such that $T(s)R(s)=Z(s)$ almost everywhere.
Taking adjoints and applying both sides to $v$ gives
\begin{equation}
F(R(s)^*)
=
R(s)^*T(s)^*v
\in\mathcal R(R(s)^*)
\end{equation}
almost everywhere, a contradiction. Hence $U$ is finite-dimensional.

With $G$ and $D$ as above, it remains to show that $\mu(D^c)=0$.
For each $k\geq1$, since $G(s)\geq0$, we have
$I+kG(s)\in GL(U)$ for every $s\in S$. Since
$s\mapsto I+kG(s)$ is measurable and takes values in $GL(U)$, and
the inversion map on $GL(U)$ is continuous in the operator norm, $
s\mapsto\bigl(I+kG(s)\bigr)^{-1}$ is measurable. Since $U$ is finite-dimensional and $G(s)\geq0$, it is not difficult
to see that
\begin{equation}
P(s)
:=
\lim_{k\to\infty}
\bigl(I+kG(s)\bigr)^{-1}
\end{equation}
exists as the orthogonal projection
$\ker G(s)=\ker R(s)$. Therefore, as a pointwise limit of measurable
$\mathcal L(U)$-valued mappings, $P:S\to\mathcal L(U)$ is measurable.
Let $\{e_1,\ldots,e_d\}$ be an orthonormal basis of $U$. Since
$P(s)\neq0$ on $D^c$,
\begin{equation}
D^c
=
\bigcup_{j=1}^d
\{s\in D^c:P(s)e_j\neq0\}.
\end{equation}
If $\mu(D^c)>0$, there exists $j$ such that
$D_j:=\{s\in D^c:P(s)e_j\neq0\}$ has positive measure. Fix $v\in V$ with $\lVert v\rVert_V=1$ and let
\begin{equation} \label{Pej}
Z(s)w
=
\langle w,P(s)e_j\rangle_U v.
\end{equation}
Then $Z\in L^2(S;L_2(U,V))$. By (1),
$T(s)R(s)=Z(s)$ almost everywhere for some measurable
$T:S\to L_2(H,V)$. Hence, for almost every $s\in D_j$,
\begin{equation}
0
=
T(s)R(s)P(s)e_j
=
Z(s)P(s)e_j
=
\lVert P(s)e_j\rVert_U^2v,
\qquad \text{a.e.},
\end{equation}
a contradiction since $P(s)e_j\neq0$ on $D_j$. Therefore
$\mu(D^c)=0$, and (2) follows.
\end{proof}
We now construct the required progressively measurable
process $T$ in the case where $U$ is finite-dimensional. First, applying \Cref{lem:measurable_factorization} with $S=[0,T]\times\Omega$, $\mathcal A=\mathcal B([0,T])\otimes\mathcal F_T$, and $\mu=ds\otimes d\mathbb Q$, we construct
$T$ according to \eqref{setd}. To establish the progressive measurability of $T$,
fix $t\in(0,T]$ and apply \Cref{lem:measurable_factorization} to the restrictions of $R$ and $Z$ to $S_t=[0,t]\times\Omega$, equipped with $\mathcal A_t=\mathcal B([0,t])\otimes\mathcal F_t$ and the restriction of $\mu$ to $S_t$. Denote the resulting mapping by $T^{(t)}$. Since the set $D$ defined in \eqref{setd1} depends only on the pointwise values of $R$, the corresponding set in $S_t$ is precisely $D\cap S_t$. Moreover, since the construction in \eqref{conTs} depends only on the pointwise values of $R$ and $Z$, we have $T^{(t)}(s,\omega)=T(s,\omega),
 (s,\omega)\in S_t$. Consequently, the restriction of $T$ to $S_t$
is strongly measurable with respect to $\mathcal A_t$.
Since $t\in(0,T]$ is arbitrary, $T$ is progressively
measurable.

On the other hand, it follows from
\Cref{lem:measurable_nonrange,lem:measurable_factorization} that the markets described by \eqref{Qstrong} and \eqref{Qmild}
are not complete in the classical sense when $U_0$ is infinite-dimensional;
that is, not every square-integrable contingent claim is strongly attainable in the former market or mildly attainable in the latter.
However, we may derive an asymptotic completeness result by considering
a sequence of markets with finite-dimensional noise. More precisely,
let $\{e_k\}_{k\geq1}$ be an orthonormal basis of $U_0$, set $
U_0^n=\operatorname{span}\{e_1,\ldots,e_n\}$, and let $P_n:U_0\to U_0^n$ denote the orthogonal projection. Setting $W^{n,\mathbb Q}:=P_nW^{\mathbb Q}$, we consider
the truncated systems corresponding to \eqref{Qmild}
and \eqref{Qstrong}, respectively:
\begin{align}
X_t^{n,x}
&=
S(t)x
+
\int_0^t S(t-s)B(X_s^{n,x})P_n\,dW_s^{\mathbb Q},
\label{Qtruncated_mild}\\
X_t^{n,x}
&=
x
+
\int_0^t AX_s^{n,x}\,ds
+
\int_0^t B(X_s^{n,x})P_n\,dW_s^{\mathbb Q}.
\label{Qtruncated_strong}
\end{align}
Suppose that the market is described by either \eqref{Qmild} or \eqref{Qstrong}, and that $\mathcal F_T$ is the $\mathbb Q$-augmentation of the sigma-algebra generated by $W^{\mathbb Q}$ up to time $T$. Assume further that, for every $n$, the truncated asset-price process $X^{n,x}$ is tradable and $B(X_t^{n,x})|_{U_0^n}$ is injective almost everywhere. 

For any $\xi\in L^2(\Omega,\mathcal F_T,\mathbb Q;H)$, let $\xi^n=\mathbb E^{\mathbb Q}[\xi\mid\mathcal F_T^n]$, where $\mathcal F_T^n$ is the $\mathbb Q$-augmentation of the sigma-algebra generated by $W^{n,\mathbb Q}$ up to time $T$. By \Cref{lem:measurable_factorization} and the replication theorem, each truncated market is complete. Thus, there exists an admissible portfolio $(\phi^n,\theta^n)$ whose value process satisfies
\begin{equation}
V_t^n=\phi_t^n X_t^{n,x}+\theta_t^n B_t,
\qquad V_T^n=\xi^n.
\end{equation}
Here, we assume that each truncated asset-price process $X^{n,x}$ is tradable, as is the case for the model considered in \Cref{GBMsec}. Since $\xi^n\to\xi$ in $L^2$ by the martingale convergence theorem, we obtain $V_T^n\to\xi$ in $L^2$. In this case, we call the market \emph{asymptotically complete}, meaning that every square-integrable contingent claim can be approximated by the terminal wealths of admissible  portfolios in the truncated markets, with convergence in $L^2(\Omega,\mathcal F_T,\mathbb Q;H)$. A similar argument applies to the real-valued case with $A=rI$, since the corresponding replication problem has essentially the same structure and is also covered by \Cref{lem:measurable_factorization}. Finally, completeness is also related to uniqueness of the risk-neutral measure.
\begin{theorem}
Suppose that the market is described by either \eqref{Qstrong}
or \eqref{Qmild}, and that $U_0$ is finite-dimensional.
Then the market is complete if and only if the equivalent
risk-neutral measure is unique\footnote{Here, uniqueness is understood within the class of
equivalent risk-neutral measures induced by the Girsanov shifts
considered in this paper.}.
\end{theorem}

\begin{proof}
By Proposition~\ref{lem:measurable_factorization}, market completeness implies
that $B(X_t)$ is injective almost everywhere. Since any Girsanov shift
$\theta$ associated with a risk-neutral measure satisfies $B(X_t)\theta_t=F(X_t)$ almost everywhere, the injectivity of $B(X_t)$ implies that $\theta$ is
unique. Hence the equivalent risk-neutral measure is unique.

Conversely, suppose that the equivalent risk-neutral measure is unique.
If $B(X_t)$ were not injective on a set of positive
$dt\otimes d\mathbb Q$-measure, select $Pe_j$ as in \eqref{Pej} such
that $Pe_j$ is nonzero on a set of positive
$dt\otimes d\mathbb Q$-measure. Since $Pe_j\in\ker B(X_t)$ and
$\lVert Pe_j\rVert_{U_0}\leq1$, the stochastic exponential
\begin{equation}
\exp\left(
\int_0^T
\left\langle Pe_j,dW_s^{\mathbb Q}\right\rangle_{U_0}
-
\frac{1}{2}
\int_0^T
\lVert Pe_j\rVert_{U_0}^2\,ds
\right)
\end{equation}
defines an equivalent probability measure $\widetilde{\mathbb Q}$.
Under $\widetilde{\mathbb Q}$, the corresponding Girsanov shift does not
change the drift of $X$, since $B(X_t)Pe_j=0$ almost everywhere. Hence
$\widetilde{\mathbb Q}$ is also risk-neutral. Moreover,
$\widetilde{\mathbb Q}\neq\mathbb Q$, contradicting uniqueness.
Therefore, $B(X_t)$ is injective almost everywhere, and
Proposition~\ref{lem:measurable_factorization} implies that the market is
complete.
\end{proof}

The finite-dimensionality of \(U_0\) is essential for the converse. Indeed, the infinite-dimensional model considered in the next section provides a counterexample in which the equivalent risk-neutral measure is unique while the market is not complete in the sense defined above. We also note that the strategies considered in this section are restricted to $L(H)$ and $H^{\prime}$. Enlarging the space of admissible strategies, for instance, by allowing unbounded strategies, may lead to different results.

\section{Markets Driven by Geometric Brownian Motions in Hilbert Spaces}\label{GBMsec}
In this section, we study a special case in which both the generator $A$ and the volatility operator have specific structures, rather than the general linear forms considered in the previous sections. This model can be viewed as a direct infinite-dimensional extension of the classical Black--Scholes model and as a market consisting of infinitely many one-dimensional risky assets. Moreover, while the existence of the risk-neutral measure in \Cref{rnm} is difficult to verify in general, the specific structure of the present model allows us to establish a more tractable sufficient condition.

\subsection{General bounded sequence $r_i$}\label{generi}
We consider the following stochastic evolution equation in the Hilbert space $H$:
\begin{equation}\label{gmbh}
X_t
=
x+\int_0^t \mathcal M X_r\,dr
+\int_0^t \Sigma X_r\,dW_r^{\mathbb Q},
\end{equation}
where $U=H$, and $W^{\mathbb Q}$ is a $Q$-Wiener process on $H$.
We assume that $Q$ admits an orthonormal basis of eigenvectors
$\{e_n\}_{n\geq1}$ with corresponding positive eigenvalues
$\{\lambda_n\}_{n\geq1}$.
The operator $\mathcal M\in\mathcal L(H)$ is diagonal with respect
to the same basis, with
\begin{equation}
\mathcal M x
=
\sum_{n=1}^{\infty}
\mu_n\langle x,e_n\rangle e_n,
\end{equation}
for some bounded sequence $(\mu_n)\subset\mathbb R$, and
$\Sigma\in\mathcal L\bigl(H,\mathcal L_2(H_Q;H)\bigr)$ is defined by
\begin{equation}
(\Sigma x)(\cdot)
=
\sum_{n=1}^{\infty}
\langle x,e_n\rangle
\langle \cdot,e_n\rangle e_n.
\end{equation}
Since $\mathcal M$ is bounded, \eqref{gmbh} admits a unique strong
solution. Writing $X_t^i:=\langle X_t,e_i\rangle$ and
$x^i:=\langle x,e_i\rangle$, taking the inner product of
\eqref{gmbh} with $e_i$ gives
\begin{equation}
X_t^i
=
x^i+\int_0^t\mu_iX_r^i\,dr
+\int_0^t\sqrt{\lambda_i}X_r^i\,d\beta_r^i.
\end{equation}
Indeed, by the definitions of $\mathcal M$ and $\Sigma$,
$\langle\mathcal M X_r,e_i\rangle=\mu_iX_r^i$ and
\begin{equation}
\langle(\Sigma X_r)h,e_i\rangle
=
X_r^i\langle h,e_i\rangle,
\qquad h\in H_Q.
\end{equation}
Thus, each coordinate process $X^i$ is a one-dimensional geometric
Brownian motion and is given explicitly by
\begin{equation}
X_t^i
=
x^i
\exp\left(
\left(\mu_i-\frac{\lambda_i}{2}\right)t
+\sqrt{\lambda_i}\,\beta_t^i
\right).
\end{equation}
Consequently,
\begin{equation}\label{eq:explicit-solution}
X_t
=
\sum_{i=1}^{\infty}
x^i
\exp\left(
\left(\mu_i-\frac{\lambda_i}{2}\right)t
+\sqrt{\lambda_i}\,\beta_t^i
\right)e_i,
\end{equation}
where, for every fixed $t\geq0$, the series converges to the unique
strong solution of \eqref{gmbh} in $H$, $\mathbb Q$-almost surely, and
in $L^2(\Omega;H)$. We further assume that the accumulation operator of the risk-neutral
market is generated by a diagonal operator $A\in\mathcal L(H)$ with
respect to the same orthonormal basis $\{e_i\}_{i\geq1}$, given by
\begin{equation}
Ax
=
\sum_{i=1}^{\infty}
r_i\langle x,e_i\rangle e_i,
\end{equation}
where $(r_i)_{i\geq1}$ is a bounded sequence of real
numbers. Consequently, the accumulation and discount operators are,
respectively,
\begin{equation}
S(t)x
=
e^{tA}x
=
\sum_{i=1}^{\infty}
e^{r_i t}\langle x,e_i\rangle e_i
\end{equation}
and
\begin{equation}
S(-t)x
=
e^{-tA}x
=
\sum_{i=1}^{\infty}
e^{-r_i t}\langle x,e_i\rangle e_i.
\end{equation}
\begin{theorem}\label{thm:rnm_gbm}
Suppose that
\begin{equation}\label{eq:rnm_condition}
\sum_{i=1}^{\infty}
\frac{(r_i-\mu_i)^2}{\lambda_i}
<\infty.
\end{equation}
Then there exists an equivalent probability measure
$\mathbb Q\sim\mathbb P$ under which
\begin{equation}
X_t
=
x+\int_0^t AX_s\,ds
+\int_0^t\Sigma X_s\,dW_s^{\mathbb Q}.
\end{equation}
Consequently, $\{S(-t)X_t\}_{t\in[0,T]}$ is an $H$-valued
$\mathbb Q$-martingale, and hence $\mathbb Q$ is a risk-neutral
measure. Moreover, if $\langle x,e_i\rangle\neq0, i\geq1$, then the risk-neutral measure is unique.
\end{theorem}
\begin{proof}
We first prove existence. Define $\theta
=
\sum_{i=1}^{\infty}
\frac{\mu_i-r_i}{\sqrt{\lambda_i}}
\bigl(\sqrt{\lambda_i}e_i\bigr)$. Since $\{\sqrt{\lambda_i}e_i\}_{i\geq1}$ is an orthonormal basis of
$H_0$, \eqref{eq:rnm_condition} yields
\begin{equation}
\lVert\theta\rVert_{H_0}^2
=
\sum_{i=1}^{\infty}
\frac{(\mu_i-r_i)^2}{\lambda_i}
<\infty.
\end{equation}
Hence, $\theta\in H_0$. Moreover, $
(\Sigma x)\theta
=
(\mathcal M-A)x, x\in H$.
By the Girsanov theorem, there exists an equivalent probability
measure $\mathbb Q\sim\mathbb P$ such that $W_t^{\mathbb Q}
=
W_t^{\mathbb P}+t\theta$ is a $Q$-Wiener process under $\mathbb Q$. Hence,
\begin{equation}
dX_t
=
\left[
\mathcal M X_t-(\Sigma X_t)\theta
\right]dt
+
\Sigma X_t\,dW_t^{\mathbb Q}
\nonumber\\
=
AX_t\,dt+\Sigma X_t\,dW_t^{\mathbb Q}.
\end{equation}
Consequently,
\begin{equation}
d\bigl(S(-t)X_t\bigr)
=
S(-t)\Sigma X_t\,dW_t^{\mathbb Q},
\end{equation}
and therefore $\{S(-t)X_t\}_{t\in[0,T]}$ is an $H$-valued
$\mathbb Q$-martingale. Hence, $\mathbb Q$ is a risk-neutral measure. We now prove uniqueness. Suppose that $\widetilde{\mathbb Q}\sim
\mathbb P$ is another risk-neutral measure, and let
$\widetilde\theta$ be its corresponding Girsanov shift. Then
\begin{equation}
dX_t
=
\left[
\mathcal M X_t-(\Sigma X_t)\widetilde\theta_t
\right]dt
+
\Sigma X_t\,dW_t^{\widetilde{\mathbb Q}}.
\end{equation}
Consequently,
\begin{equation}
d\bigl(S(-t)X_t\bigr)
=S(-t)\left[\mathcal M X_t-AX_t-(\Sigma X_t)\widetilde\theta_t\right]dt
+S(-t)\Sigma X_t\,dW_t^{\widetilde{\mathbb Q}}.
\end{equation}
Since $\{S(-t)X_t\}_{t\in[0,T]}$ is a
$\widetilde{\mathbb Q}$-martingale, its finite-variation part must
vanish. Hence,
\begin{equation}
\mathcal M X_t-AX_t-(\Sigma X_t)\widetilde\theta_t=0
\qquad
dt\otimes d\mathbb P\text{-a.e.},
\end{equation}
where we have used $\widetilde{\mathbb Q}\sim\mathbb P$. Therefore,
\begin{equation}
X_t^i
\left(
\mu_i-r_i-\langle\widetilde\theta_t,e_i\rangle_H
\right)
=0,
\qquad i\geq1,
\end{equation}
for $dt\otimes d\mathbb P$-almost every $(t,\omega)$.

If $\langle x,e_i\rangle\neq0$ for every $i\geq1$, the explicit
representation of $X^i$ implies that $X_t^i\neq0$ for every $t$,
almost surely. Hence,
\begin{equation}
\frac{\langle\widetilde\theta_t,e_i\rangle_H}{\sqrt{\lambda_i}}
=
\left\langle
\widetilde\theta_t,\sqrt{\lambda_i}e_i
\right\rangle_{H_0}
=
\frac{\mu_i-r_i}{\sqrt{\lambda_i}},
\qquad i\geq1,
\end{equation}
for $dt\otimes d\mathbb P$-almost every $(t,\omega)$. Therefore $\widetilde\theta_t=\theta_t, dt\otimes d\mathbb P$-almost every.
\end{proof}
In the special case $A=rI$, that is, $r_i=r$ for every $i\geq1$,
condition \eqref{eq:rnm_condition} reduces to
\begin{equation}\label{risknr}
\sum_{i=1}^{\infty}
\frac{(r-\mu_i)^2}{\lambda_i}
<\infty.
\end{equation}
For example, the condition above is satisfied by taking
\begin{equation}
\lambda_i=\frac{1}{i^2},
\qquad
\mu_i=r+\frac{1}{i^2},
\qquad i\geq1.
\end{equation}
Indeed, $Q$ is trace class and
\begin{equation}
\sum_{i=1}^{\infty}
\frac{(r-\mu_i)^2}{\lambda_i}
=
\sum_{i=1}^{\infty}\frac{1}{i^2}
<\infty.
\end{equation}
Condition \eqref{risknr} is more restrictive than
\eqref{eq:rnm_condition}, where the sequence $(r_i)_{i\geq1}$
is allowed to vary. Indeed, for any given sequences
$(\mu_i)_{i\geq1}$ and $(\lambda_i)_{i\geq1}$ with
$\lambda_i>0$, one can choose $ r_i=\mu_i+\sqrt{\lambda_i}\,a_i,
\quad i\geq1$, where $(a_i)_{i\geq1}\in\ell^2$, so that
\eqref{eq:rnm_condition} is satisfied. In contrast,
there may be no single constant $r$ satisfying
\eqref{risknr}. For example, consider
\begin{equation}
\lambda_i=\frac{1}{i^2},
\qquad
\mu_i=\frac{1}{i},
\qquad i\geq1.
\end{equation}
Then $Q$ is trace class, but
\begin{equation}
\sum_{i=1}^{\infty}
\frac{(r-\mu_i)^2}{\lambda_i}
=
\sum_{i=1}^{\infty}(ri-1)^2
=\infty
\end{equation}
for every $r\in\mathbb R$. We next derive the Black--Scholes equation for this particular case. By taking advantage of the explicit representation of \(X_t\) in \eqref{eq:explicit-solution}, we can relax the smoothness assumptions on the payoff function \(G\) imposed in \Cref{ass:pricing_regular}. To this end, we first establish the following differentiation-under-the-integral result.
\begin{proposition}\label{diunin}
Let $(S,\mathcal A,\mu)$ be a finite measure space, and let $X$ and $Y$
be separable Banach spaces. Let $\mathcal{O}$ be an open subset of $X$, and let
\begin{equation}
f:\mathcal{O}\times S\to Y
\end{equation}
be a function. Assume:
\begin{enumerate}[(i)]
\item For every $x\in \mathcal{O}$, the map $s\mapsto f(x,s)$ is measurable and
Bochner integrable over $S$.

\item For every $x\in\mathcal O$, the map $y\mapsto f(y,s)$
is twice Fr\'echet differentiable at $x$ for $\mu$-almost
every $s\in S$, that is, there exists a measurable set
$B_x\in\mathcal A$ of full measure such that
$y\mapsto f(y,s)$ is twice Fr\'echet differentiable at $x$
for every $s\in B_x$, with
\begin{equation}
f'(x,s)\in\mathcal L(X,Y),
\qquad
f''(x,s)\in\mathcal L\bigl(X,\mathcal L(X,Y)\bigr).
\end{equation}

\item  There exist $\ell^1,\ell^2\in L^1(S,\mu)$
such that, for every $x\in\mathcal O$,
\begin{equation}
\lVert f'(x,s)\rVert_{\mathcal L(X,Y)}
\leq \ell^1(s),
\end{equation}
and
\begin{equation}
\lVert f''(x,s)\rVert_{\mathcal L(X,\mathcal L(X,Y))}
\leq \ell^2(s),
\end{equation}
whenever the corresponding derivatives are defined.

\item For every $x\in \mathcal{O}$ and every nonzero sequence $\{h_n\}$ in $X$
with $x+h_n\in \mathcal{O}, h_n\to0$, there exist
$\rho^1_{x,\{h_n\}},\rho^2_{x,\{h_n\}}\in L^1(S,\mu)$ such that
\begin{equation}
\frac{
\left\lVert
f(x+h_n,s)-f(x,s)-f'(x,s)h_n
\right\rVert_Y
}{
\left\lVert h_n\right\rVert_X
}
\leq
\rho^1_{x,\{h_n\}}(s),
\qquad s\in B_x,
\end{equation}
and
\begin{equation}
\frac{
\left\lVert
f'(x+h_n,s)-f'(x,s)-f''(x,s)h_n
\right\rVert_{\mathcal L(X,Y)}
}{
\left\lVert h_n\right\rVert_X
}
\leq
\rho^2_{x,\{h_n\}}(s),
\qquad
s\in B_x\cap B_{x+h_n},
\end{equation}
for every $n\geq1$.
\end{enumerate}

Then:
\begin{itemize}
\item For every $x\in \mathcal{O}$, there exists an equivalence class of
$\sigma_{SO}(X,Y)$-measurable functions $\bar f_x(s)$ such that
$\bar f_x(s)z$ is integrable over $S$ for all $z\in X$. Moreover, the
mapping $F:\mathcal{O}\to Y$, defined by
\begin{equation}
F(x)=\int_S f(x,s)\,\mu(ds),
\end{equation}
is Fr\'echet differentiable on $\mathcal{O}$ with
$F'(x)\in\mathcal L(X,Y)$ defined by
\begin{equation}
F'(x)z
=
\int_S\bar f_x(s)z\,\mu(ds),
\qquad z\in X.
\end{equation}

\item For every $x\in \mathcal{O}$, there exists an equivalence class of functions
$\bar{\bar f}_x:S\to\mathcal L\bigl(X,\mathcal L(X,Y)\bigr)$ such that,
for all $z_1,z_2\in X$, $\bar{\bar f}_x(s)z_1$ is
$\sigma_{SO}(X,Y)$-measurable and
$(\bar{\bar f}_x(s)z_1)z_2$ is integrable over $S$. Moreover, the
mapping $F$ is twice Fr\'echet differentiable on $\mathcal{O}$ with
$F''(x)\in\mathcal L\bigl(X,\mathcal L(X,Y)\bigr)$ defined by
\begin{equation}
F''(x)z_1
=
\int_S\bar{\bar f}_x(s)z_1\,\mu(ds),
\qquad z_1\in X,
\end{equation}
and
\begin{equation}
(F''(x)z_1)z_2
=
\int_S(\bar{\bar f}_x(s)z_1)z_2\,\mu(ds),
\qquad z_1,z_2\in X.
\end{equation}
\end{itemize}
\end{proposition}

\begin{proof}
For any $x\in \mathcal{O}$, we first show that such an equivalence class exists.
Let $g_x:S\to\mathcal L(X,Y)$ be any
$\sigma_{SO}(X,Y)$-measurable function and define
\begin{equation}
\bar f_x(s)
=
f'(x,s)\mathbb I_{B_x}(s)
+
g_x(s)\mathbb I_{B_x^c}(s).
\end{equation}
The measurability of $\bar f_x$ follows from
\begin{equation}
\bar f_x(s)h
=
\lim_{n\to\infty}
\frac{f(x+\varepsilon_n h,s)-f(x,s)}{\varepsilon_n}
\mathbb I_{B_x}(s)
+
g_x(s)h\mathbb I_{B_x^c}(s),
\qquad h\in X,\ s\in S,
\end{equation}
where $(\varepsilon_n)_{n\geq1}\subset\mathbb Q\setminus\{0\}$ is a
sequence with $\varepsilon_n\to0$. Indeed,
\begin{equation}
\frac{f(x+\varepsilon_n h,s)-f(x,s)}{\varepsilon_n}
\mathbb I_{B_x}(s)
+
g_x(s)h\mathbb I_{B_x^c}(s)
\end{equation}
is a sequence of measurable $Y$-valued functions. Hence,
$\bar f_x(s)h$ is a measurable $Y$-valued function for every
$h\in X$, and therefore $\bar f_x$ is
$\sigma_{SO}(X,Y)$-measurable.

Then, choose an arbitrary nonzero sequence $\{h_n\}$ in $X$ such that
$h_n\to0$. We observe that
\begin{align}\label{1frecon}
\frac{\left\lVert F(x+h_n)-F(x)-\int_S\bar f_x(s)h_n\,\mu(ds)\right\rVert}
{\left\lVert h_n\right\rVert}
={}&
\frac{\left\lVert\int_S
\bigl(f(x+h_n,s)-f(x,s)-\bar f_x(s)h_n\bigr)\,\mu(ds)\right\rVert}
{\left\lVert h_n\right\rVert}
\notag\\
\leq{}&
\int_S
\frac{\left\lVert f(x+h_n,s)-f(x,s)-\bar f_x(s)h_n\right\rVert}
{\left\lVert h_n\right\rVert}
\,\mu(ds).
\end{align}
The sequence of functions
\begin{equation}
\frac{\left\lVert
f(x+h_n,s)-f(x,s)-\bar f_x(s)h_n
\right\rVert}
{\left\lVert h_n\right\rVert}
\end{equation}
converges to zero on $B_x$ and is dominated there by
$\rho^1_{x,\{h_n\}}\in L^1(S,\mu)$. Since $B_x$
has full measure and $\mu(S)<\infty$, \eqref{1frecon} converges to zero
by the dominated convergence theorem.

We now move to the second part of the proposition. For each $y\in \mathcal{O}$,
we take the version of $\bar f_y$ corresponding to $g_y=0$. Let
$\{e_i\}_{i\geq1}$ be a countable dense subset of $X$. For each
$i\geq1$, let $(\varepsilon_n)_{n\geq1}\subset\mathbb Q\setminus\{0\}$
be a sequence converging to zero and define
\begin{equation}
D_{x,i}
=
B_x\cap\bigcap_{n=1}^{\infty}B_{x+\varepsilon_n e_i}.
\end{equation}
We set $B_{x+\varepsilon_n e_i}=S$ whenever
$x+\varepsilon_n e_i\notin\mathcal O$. Then $D_{x,i}\in\mathcal A$ and has full measure. For every
$s\in D_{x,i}$,
\begin{equation}
\frac{
\bar f_{x+\varepsilon_n e_i}(s)-\bar f_x(s)
}{
\varepsilon_n
}
\longrightarrow
f''(x,s)e_i
\end{equation}
in the uniform operator topology of $\mathcal L(X,Y)$. Hence,
$s\mapsto f''(x,s)e_i$, restricted to $D_{x,i}$, is
$\sigma_{SO}(X,Y)$-measurable. Set
\begin{equation}
D_x
=
\bigcap_{i=1}^{\infty}D_{x,i}.
\end{equation}
Then $D_x\in\mathcal A$ and has full measure. Since
$\{e_i\}_{i\geq1}$ is dense in $X$, it follows that, for every $h\in X$,
$s\mapsto f''(x,s)h$, restricted to $D_x$, is
$\sigma_{SO}(X,Y)$-measurable. Let
$G':S\to\mathcal L\bigl(X,\mathcal L(X,Y)\bigr)$ be such that
$G'(s)h$ is $\sigma_{SO}(X,Y)$-measurable for every $h\in X$, and
define
\begin{equation}
\bar{\bar f}_x(s)
=
f''(x,s)\mathbb I_{D_x}(s)
+
G'(s)\mathbb I_{D_x^c}(s).
\end{equation}
Then, for every $h\in X$, $\bar{\bar f}_x(s)h$ is
$\sigma_{SO}(X,Y)$-measurable. Again, choose an arbitrary nonzero sequence $\{h_n\}$ in $X$ such that
$h_n\to0$, and observe that
\begin{equation}
F'(x+h_n)-F'(x)
-\int_S\bar{\bar f}_x(s)h_n\,\mu(ds)
=
\int_S
\bigl(
\bar f_{x+h_n}(s)-\bar f_x(s)-\bar{\bar f}_x(s)h_n
\bigr)\,\mu(ds).
\end{equation}
Since $X$ is separable and
$\bar f_{x+h_n}$, $\bar f_x$, and $\bar{\bar f}_x(\cdot)h_n$ are
$\sigma_{SO}(X,Y)$-measurable, the function
\begin{equation}
s\mapsto
\left\lVert
\bar f_{x+h_n}(s)-\bar f_x(s)-\bar{\bar f}_x(s)h_n
\right\rVert_{\mathcal L(X,Y)}
\end{equation}
is measurable. Moreover,
\begin{equation}
\left\lVert
F'(x+h_n)-F'(x)
-\int_S\bar{\bar f}_x(s)h_n\,\mu(ds)
\right\rVert_{\mathcal L(X,Y)}
\leq
\int_S
\left\lVert
\bar f_{x+h_n}(s)-\bar f_x(s)-\bar{\bar f}_x(s)h_n
\right\rVert_{\mathcal L(X,Y)}
\,\mu(ds).
\end{equation}
Hence,
\begin{equation}
\frac{
\left\lVert
F'(x+h_n)-F'(x)
-\int_S\bar{\bar f}_x(s)h_n\,\mu(ds)
\right\rVert_{\mathcal L(X,Y)}
}{
\left\lVert h_n\right\rVert
}
\leq
\int_S
\frac{
\left\lVert
\bar f_{x+h_n}(s)-\bar f_x(s)-\bar{\bar f}_x(s)h_n
\right\rVert_{\mathcal L(X,Y)}
}{
\left\lVert h_n\right\rVert
}
\,\mu(ds).
\end{equation}
The integrand converges to zero on
$D_x\cap\bigcap_{n=1}^{\infty}B_{x+h_n}$, which has full measure, and
is dominated there by
$\rho^2_{x,\{h_n\}}\in L^1(S,\mu)$.Since $\mu(S)<\infty$, the conclusion follows
from the dominated convergence theorem.
\end{proof}
\begin{remark}
The exceptional sets in Proposition~\ref{diunin} may depend on the
point of differentiation in a nontrivial way. To illustrate this, let
$X=S$ be a separable Banach space equipped with a finite nonatomic
Borel measure $\mu$, and fix $h\in X\setminus\{0\}$. For each
$s\in X$, consider the countable set
$A_s=\{s+ih:i\in\mathbb Z\}$. Since
$\lVert (s+ih)-(s+jh)\rVert_X=|i-j|\lVert h\rVert_X$, the set $A_s$
has no accumulation points.

Let $\phi:X\times X\to Y$ be a function satisfying the differentiability
and domination conditions of Proposition~\ref{diunin}, and define
$f:X\times X\to Y$ by setting $f(x,s)=\phi(x,s)$ for $x\notin A_s$,
while allowing the values of $f(x,s)$ to be modified arbitrarily when
$x\in A_s$. For each fixed $x\in \mathcal{O}$, the set of $s$ for which $x$ is
one of these exceptional points is
\begin{equation}
N_x
=
\{s\in X:x\in A_s\}
=
\{x-ih:i\in\mathbb Z\}.
\end{equation}
Thus $N_x$ is countable and, since $\mu$ is nonatomic,
$\mu(N_x)=0$. Moreover, for every $s\notin N_x$, the point $x$ does
not belong to $A_s$, and since $A_s$ has no accumulation points,
$f(\cdot,s)$ agrees with $\phi(\cdot,s)$ on a neighborhood of $x$.
Consequently,
$f'(x,s)=\phi'(x,s)$ and $f''(x,s)=\phi''(x,s)$ for $s\notin N_x$.

Furthermore, for any sequence $\{u_n\}$ in $X$ with $u_n\to0$, the
set $N_x\cup\bigcup_{n=1}^{\infty}N_{x+u_n}$ is countable and hence
$\mu$-null. Outside this set, all the terms appearing in the first-
and second-order difference quotients in Proposition~\ref{diunin}
coincide with the corresponding terms for $\phi$. Hence, provided the
required integrable bounds hold for $\phi$, they also hold for $f$.
This shows that Proposition~\ref{diunin} allows the map
$x\mapsto f(x,s)$ to have countably many isolated discontinuities for
every $s$, even though these exceptional points may vary with $s$.
\end{remark}
We now use \Cref{diunin}, together with the explicit representation of $X_t$, to derive the Kolmogorov equation by an alternative method to that of \Cref{lem:kolmogorov}. An advantage of this approach is that the smoothness assumptions on the payoff function $G$ can be weakened. We first impose the following conditions.
\begin{assumption} \label{assg1}
The function $G:H\to H$ is Borel measurable and satisfies
the following conditions:
\begin{enumerate}
    \item There exists a constant $C>0$ such that
\[
\|G(x)\|_H\le C\bigl(1+\|x\|_H^2\bigr),\qquad x\in H.
\]

    \item For every $t\in[0,T]$, $x\in H$, and
    $N\in\mathbb N$, $G$ is twice Fr\'echet differentiable
    at
    \begin{equation} \label{mappingg}
    \sum_{i=1}^{\infty}
    \langle x,e_i\rangle
    \exp\left(
        \left(r_i-\frac{\lambda_i}{2}\right)t
        +\langle h,e_i\rangle
    \right)e_i
    \end{equation}
    for $\mathcal L(W(t))$-almost every $h\in H$, and at
    \begin{equation}
    \sum_{i=1}^{N}
    \langle x,e_i\rangle
    \exp\left(
        \left(r_i-\frac{\lambda_i}{2}\right)t
        +\sqrt{\lambda_i}\,z_i
    \right)e_i
    \end{equation}
    for $\mathcal L(\beta_1(t),\ldots,\beta_N(t))$-almost
    every $(z_1,\ldots,z_N)\in\mathbb R^N$.

    \item There exist constants $L_0,L_1,L_2>0$ such that
    \begin{align*}
    \|G(y)-G(z)\|_H
    &\le L_0\|y-z\|_H,\\
    \|G'(y)-G'(z)\|_{\mathcal L(H)}
    &\le L_1\|y-z\|_H,\\
    \|G''(y)-G''(z)\|_{\mathcal L(H,\mathcal L(H))}
    &\le L_2\|y-z\|_H,
    \end{align*}
    where the first inequality holds for all $y,z\in H$,
    and the second and third inequalities hold whenever
    the respective derivatives exist at both $y$ and $z$.
\end{enumerate}
\end{assumption}
Define the mappings $\Phi,\Phi^N:[0,T]\times H\times H\to H$,
$N\in\mathbb N$, by
\begin{equation}
\Phi(t,x,h)
:=
\sum_{i=1}^{\infty}
\langle x,e_i\rangle
e^{(r_i-\lambda_i/2)t+\langle h,e_i\rangle}e_i,
\qquad
\Phi^N(t,x,h)
:=
\sum_{i=1}^N
\langle x,e_i\rangle
e^{(r_i-\lambda_i/2)t+\langle h,e_i\rangle}e_i.
\end{equation}
Under \Cref{assg1}, the mappings
\[
(t,x,h)\mapsto G(\Phi(t,x,h))
\qquad\text{and}\qquad
(t,x,h)\mapsto G(\Phi^N(t,x,h))
\]
satisfy Conditions {\rm (i)--(iv)} of \Cref{diunin}
with respect to $\mathcal L(W(t))$ for every $t\in[0,T]$.
By the definition of the pushforward measures, the
corresponding mappings on the original probability space,
given by
\begin{equation}
(t,x,\omega)\mapsto G(\Phi(t,x,W(t,\omega)))
\end{equation}
and
\begin{equation}
(t,x,\omega)\mapsto G(\Phi^N(t,x,W(t,\omega))),
\end{equation}
also satisfy Conditions {\rm (i)--(iv)} of \Cref{diunin}
with respect to $\mathbb P$.

Moreover, whenever $G$ is Fr\'echet differentiable at
$\Phi(t,x,h)$, the chain rule yields the derivative of
$G(\Phi(t,x,h))$ with respect to $x$, applied to $k\in H$,
as $G'(\Phi(t,x,h))
\bigl(\Phi(t,k,h)\bigr)$. Similarly, its partial derivative with respect to $t$
is given by
\begin{equation}
G'(\Phi(t,x,h))
\left(
\sum_{i=1}^{\infty}
\left(r_i-\frac{\lambda_i}{2}\right)
\langle x,e_i\rangle
e^{(r_i-\lambda_i/2)t+\langle h,e_i\rangle}e_i
\right).
\end{equation}
Whenever $G$ is twice Fr\'echet differentiable at
$\Phi(t,x,h)$, the second derivative of
$G(\Phi(t,x,h))$ with respect to $x$, applied to
$k_1,k_2\in H$, is given by
\begin{equation}
G''(\Phi(t,x,h))
\bigl(\Phi(t,k_1,h),\Phi(t,k_2,h)\bigr).
\end{equation}
The corresponding identities for the finite-dimensional
mapping and for the mappings on the original probability
space follow in the same manner.
\begin{theorem}
Define
\begin{equation}
u(t,x):=\mathbb E^{\mathbb Q}\left[G(X_t^x)\right],
\qquad (t,x)\in[0,T]\times H.
\end{equation}
Then, for every $t\in (0,T)$ and $x\in H$,
\begin{equation}
\partial_tu(t,x)
=
\partial_xu(t,x)(Ax)
+
\frac12
\sum_{i=1}^{\infty}
\lambda_i\langle x,e_i\rangle^2
\bigl((\partial_{xx}u(t,x)e_i)e_i\bigr),
\end{equation}
with initial condition $u(0,x)=G(x)$.
\end{theorem}

\begin{proof}
For every $N\in\mathbb N$, define the mapping
$\Psi^N:[0,T]\times H\times\mathbb R^N\to H$ by
\begin{equation}
\Psi^N(t,x,z)
:=
\sum_{i=1}^N
\langle x,e_i\rangle
\exp\!\left(
\left(r_i-\frac{\lambda_i}{2}\right)t
+\sqrt{\lambda_i t}\,z_i
\right)e_i,
\end{equation}
and
\begin{equation}
u^N(t,x)
:=
\int_{\mathbb R^N}
G\bigl(\Psi^N(t,x,z)\bigr)\,\gamma_N(dz)
=
\mathbb E\left[
G\bigl(\Phi^N(t,x,W(t))\bigr)
\right],
\end{equation}
where
\begin{equation}
\gamma_N(dz)
:=
\prod_{i=1}^N
\frac{1}{\sqrt{2\pi}}e^{-z_i^2/2}\,dz_i.
\end{equation}
We first show that
\begin{equation}
\partial_tu^N(t,x)
=
\partial_xu^N(t,x)(A_Nx)
+
\frac12
\sum_{i=1}^N
\lambda_i\langle x,e_i\rangle^2
\bigl((\partial_{xx}u^N(t,x)e_i)e_i\bigr),
\end{equation}
where
\begin{equation}
A_Nx:=\sum_{i=1}^N r_i\langle x,e_i\rangle e_i.
\end{equation}

Differentiating under the integral gives
\begin{align}\label{VN1}
\partial_tu^N(t,x)
={}&
\int_{\mathbb R^N}
G'\bigl(\Psi^N(t,x,z)\bigr)
\sum_{i=1}^N
\left(r_i-\frac{\lambda_i}{2}\right)
\langle x,e_i\rangle
e^{(r_i-\lambda_i/2)t+\sqrt{\lambda_i t}\,z_i}
e_i\,\gamma_N(dz)
\nonumber\\
&+
\frac{1}{2t}
\sum_{i=1}^N
\int_{\mathbb R^N}
z_i\sqrt{\lambda_i t}\,
G'\bigl(\Psi^N(t,x,z)\bigr)
\langle x,e_i\rangle
e^{(r_i-\lambda_i/2)t+\sqrt{\lambda_i t}\,z_i}
e_i\,\gamma_N(dz).
\end{align}
For each $i$, assume that, for a.e.\ $z_{-i}\in\mathbb R^{N-1}$, the map
\begin{equation}
z_i\mapsto
G'\bigl(\Psi^N(t,x,z)\bigr)
\left(
\langle x,e_i\rangle
e^{(r_i-\lambda_i/2)t+\sqrt{\lambda_i t}\,z_i}e_i
\right)
\frac{1}{\sqrt{2\pi}}e^{-z_i^2/2},
\end{equation}
where defined, admits an absolutely continuous version. Since $G'$ is
bounded, the absolutely continuous version vanishes as
$z_i\to\pm\infty$. Since $H$ has the Radon--Nikodým property,
the fundamental theorem of calculus remains valid
(see \cite{diestel1977vector}) for Banach-space-valued
absolutely continuous functions and gives
\begin{align}\label{VN2}
&
\int_{\mathbb R^N}
z_i\sqrt{\lambda_i t}\,
G'\bigl(\Psi^N(t,x,z)\bigr)
\langle x,e_i\rangle
e^{(r_i-\lambda_i/2)t+\sqrt{\lambda_i t}\,z_i}
e_i\,\gamma_N(dz)
\nonumber\\
={}&
\lambda_i t
\int_{\mathbb R^N}
G''\bigl(\Psi^N(t,x,z)\bigr)
\left(
\langle x,e_i\rangle
e^{(r_i-\lambda_i/2)t+\sqrt{\lambda_i t}\,z_i}e_i,
\langle x,e_i\rangle
e^{(r_i-\lambda_i/2)t+\sqrt{\lambda_i t}\,z_i}e_i
\right)
\gamma_N(dz)
\nonumber\\
&+
\lambda_i t
\int_{\mathbb R^N}
G'\bigl(\Psi^N(t,x,z)\bigr)
\langle x,e_i\rangle
e^{(r_i-\lambda_i/2)t+\sqrt{\lambda_i t}\,z_i}
e_i\,\gamma_N(dz).
\end{align}
Combining \eqref{VN1} and \eqref{VN2}, we obtain
\begin{align}
\partial_tu^N(t,x)
={}&
\int_{\mathbb R^N}
G'\bigl(\Psi^N(t,x,z)\bigr)
\sum_{i=1}^N
r_i\langle x,e_i\rangle
e^{(r_i-\lambda_i/2)t+\sqrt{\lambda_i t}\,z_i}
e_i\,\gamma_N(dz)
\nonumber\\
&+
\frac12
\sum_{i=1}^N\lambda_i
\int_{\mathbb R^N}
G''\bigl(\Psi^N(t,x,z)\bigr)
\left(
\langle x,e_i\rangle
e^{(r_i-\lambda_i/2)t+\sqrt{\lambda_i t}\,z_i}e_i,
\langle x,e_i\rangle
e^{(r_i-\lambda_i/2)t+\sqrt{\lambda_i t}\,z_i}e_i
\right)
\gamma_N(dz).
\end{align}
Moreover,
\begin{equation}
\partial_xu^N(t,x)(A_Nx)
=
\int_{\mathbb R^N}
G'\bigl(\Psi^N(t,x,z)\bigr)
\sum_{i=1}^N
r_i\langle x,e_i\rangle
e^{(r_i-\lambda_i/2)t+\sqrt{\lambda_i t}\,z_i}
e_i\,\gamma_N(dz),
\end{equation}
while, for $1\le i\le N$,
\begin{equation}
\bigl((\partial_{xx}u^N(t,x)e_i)e_i\bigr)
=
\int_{\mathbb R^N}
G''\bigl(\Psi^N(t,x,z)\bigr)
\left(
e^{(r_i-\lambda_i/2)t+\sqrt{\lambda_i t}\,z_i}e_i,
e^{(r_i-\lambda_i/2)t+\sqrt{\lambda_i t}\,z_i}e_i
\right)
\gamma_N(dz).
\end{equation}
Consequently,
\begin{equation}\label{pdevn}
\partial_tu^N(t,x)
=
\partial_xu^N(t,x)(A_Nx)
+
\frac12
\sum_{i=1}^N
\lambda_i\langle x,e_i\rangle^2
\bigl((\partial_{xx}u^N(t,x)e_i)e_i\bigr).
\end{equation}

Under \Cref{assg1}, it is not difficult to check that
\begin{equation}
u^N(t,x)
=
\mathbb E\left[G\bigl(\Phi^N(t,x,W(t))\bigr)\right]
\longrightarrow
\mathbb E\left[G\bigl(\Phi(t,x,W(t))\bigr)\right]
=
u(t,x),
\end{equation}
uniformly for $t\in[0,T]$. Moreover, for every $h\in H$,
\begin{align}
\partial_xu^N(t,x)h
&=
\mathbb E\left[
G'\bigl(\Phi^N(t,x,W(t))\bigr)
\bigl(\partial_x\Phi^N(t,x,W(t))h\bigr)
\right]
\nonumber\\
&\longrightarrow
\mathbb E\left[
G'\bigl(\Phi(t,x,W(t))\bigr)
\bigl(\partial_x\Phi(t,x,W(t))h\bigr)
\right]
=
\partial_xu(t,x)h,
\end{align}
uniformly for $t\in[0,T]$. Since $A_Nx\to Ax$, this yields
\begin{equation}
\partial_xu^N(t,x)(A_Nx)
\longrightarrow
\partial_xu(t,x)(Ax)
\end{equation}
uniformly for $t\in[0,T]$. Similarly, for every fixed
$i\in\mathbb N$ and $N\ge i$,
\begin{align}
\bigl((\partial_{xx}u^N(t,x)e_i)e_i\bigr)
&=
\mathbb E\left[
G''\bigl(\Phi^N(t,x,W(t))\bigr)
\left(
e^{(r_i-\lambda_i/2)t+\langle W(t),e_i\rangle}e_i,
e^{(r_i-\lambda_i/2)t+\langle W(t),e_i\rangle}e_i
\right)
\right]
\nonumber\\
&\longrightarrow
\mathbb E\left[
G''\bigl(\Phi(t,x,W(t))\bigr)
\left(
e^{(r_i-\lambda_i/2)t+\langle W(t),e_i\rangle}e_i,
e^{(r_i-\lambda_i/2)t+\langle W(t),e_i\rangle}e_i
\right)
\right]
\nonumber\\
&=
\bigl((\partial_{xx}u(t,x)e_i)e_i\bigr),
\end{align}
uniformly for $t\in[0,T]$. Consequently,
\begin{equation}
\sum_{i=1}^N
\lambda_i\langle x,e_i\rangle^2
\bigl((\partial_{xx}u^N(t,x)e_i)e_i\bigr)
\longrightarrow
\sum_{i=1}^\infty
\lambda_i\langle x,e_i\rangle^2
\bigl((\partial_{xx}u(t,x)e_i)e_i\bigr)
\end{equation}
uniformly for $t\in[0,T]$. Note that, by \Cref{assg1},
the series on the right-hand side converges absolutely
in $H$, uniformly for $t\in[0,T]$. Taking the limit on both sides of
\eqref{pdevn} as $N\to\infty$ therefore gives
\begin{equation}
\partial_tu(t,x)
=
\partial_xu(t,x)(Ax)
+
\frac12
\sum_{i=1}^{\infty}
\lambda_i\langle x,e_i\rangle^2
\bigl((\partial_{xx}u(t,x)e_i)e_i\bigr).
\end{equation}
\end{proof}

We point out that the method used in \Cref{lem:kolmogorov} applies equally to the present setting, while conditions~(3) and~(4) in \Cref{ass:pricing_regular} are satisfied due to the diagonal structure. The corresponding infinite-dimensional Black--Scholes equation is therefore
\begin{equation}\label{gbmbs}
\partial_tV(t,x)
+
\partial_xV(t,x)(Ax)
+
\frac12
\sum_{i=1}^{\infty}
\lambda_i\langle x,e_i\rangle^2
\bigl((\partial_{xx}V(t,x)e_i)e_i\bigr)
-AV(t,x)
=0.
\end{equation}
The regularity assumptions on $G$ provide sufficient conditions for
deriving \eqref{gbmbs}, but they are not necessary. We next consider a
case in which \eqref{gbmbs} can be obtained without imposing such
regularity on the payoff function. We refer to this setting as the
diagonal form.  We say that the payoff function $G:H\to H$ is of diagonal form if
\begin{equation}
G(x)
=
\sum_{i=1}^{\infty}
G_i\bigl(\langle x,e_i\rangle\bigr)e_i,
\qquad x\in H,
\end{equation}
where $G_i:\mathbb R\to\mathbb R$ satisfy
\begin{equation}
\sum_{i=1}^{\infty}
\left|
G_i\bigl(\langle x,e_i\rangle\bigr)
\right|^2
<\infty,
\qquad x\in H.
\end{equation}
Then the value function is also of diagonal form and can be written as
\begin{equation}
V(t,x)
=
S(-(T-t))
\mathbb E^{\mathbb Q}
\left[G\left(X_T^{t,x}\right)\right]
=
\sum_{i=1}^{\infty}
V_i\bigl(t,\langle x,e_i\rangle\bigr)e_i,
\end{equation}
where
\begin{equation}
V_i(t,z)
=
e^{-r_i(T-t)}
\mathbb E^{\mathbb Q}
\left[
G_i\left(
z\exp\left(
\left(r_i-\frac{\lambda_i}{2}\right)(T-t)
+
\sqrt{\lambda_i}\bigl(\beta_i(T)-\beta_i(t)\bigr)
\right)
\right)
\right].
\end{equation}
From a financial perspective, this diagonal structure can be
interpreted as a sequence of completely independent and segmented
markets, each with its own asset-price process and payoff function. Suppose that, for every $i\geq1$, $V_i$ is a classical solution of the
corresponding one-dimensional Black--Scholes equation,
\begin{equation}
\partial_tV_i(t,z)
+
r_i z\,\partial_zV_i(t,z)
+
\frac{\lambda_i}{2}z^2\partial_{zz}V_i(t,z)
-
r_iV_i(t,z)
=
0.
\end{equation}
If, for every $x\in H$, the sequences
\begin{equation}
\left\{
r_i x_i\partial_zV_i(t,x_i)
\right\}_{i\geq1},
\qquad
\left\{
\lambda_i x_i^2\partial_{zz}V_i(t,x_i)
\right\}_{i\geq1},
\qquad
\left\{
r_iV_i(t,x_i)
\right\}_{i\geq1},
\end{equation}
are square summable uniformly in $t\in[0,T)$, where
$x_i=\langle x,e_i\rangle$, then the one-dimensional Black--Scholes
equations can be summed in $H$, yielding \eqref{gbmbs}.

To construct infinite-dimensional analogues of classical European
call options, we first introduce an order structure on $H$.
Specifically, we equip $H$ with the $Q$-order defined by
\footnote{With this order, $H$ is an order-continuous Banach lattice
and is therefore Dedekind complete, as it is lattice isomorphic
to $\ell^2$.}
\begin{equation}
x\leq_Q y
\quad\Longleftrightarrow\quad
\langle x,e_i\rangle_H
\leq
\langle y,e_i\rangle_H,
\qquad i\geq1.
\end{equation}

\begin{example}\label{callop1}
Consider a European call option with payoff
\begin{equation}\label{call1}
G(x)=(x-K)^+,
\qquad x\in H,
\end{equation}
where $K\in H_+\setminus\{0\}$ and
$K_i:=\langle K,e_i\rangle>0$ for every $i\geq1$.
For $x_i:=\langle x,e_i\rangle>0$ and $t<T$, define
\begin{equation}
d_{1,i}
=
\frac{
\log(x_i/K_i)
+
\left(r_i+\lambda_i/2\right)(T-t)
}{
\sqrt{\lambda_i(T-t)}
},
\qquad
d_{2,i}
=
d_{1,i}-\sqrt{\lambda_i(T-t)}.
\end{equation}
The value function admits the closed-form representation
\begin{equation}
V(t,x)
=
S(t-T)\mathbb E^{\mathbb Q}
\left[G\left(X_T^{t,x}\right)\right]
=
\sum_{i=1}^{\infty}
\left[
x_i\Phi(d_{1,i})
-
K_i e^{-r_i(T-t)}\Phi(d_{2,i})
\right]e_i,
\end{equation}
where the series converges in $H$. Moreover, $V$ satisfies
\eqref{wbsh} for every $t<T$.

Indeed, by the diagonal form of the payoff and the definition of
$S(t-T)$,
\begin{equation}
V(t,x)
=
\sum_{i=1}^{\infty}
e^{-r_i(T-t)}
\mathbb E^{\mathbb Q}
\left[
\left(X_T^{i,t,x_i}-K_i\right)^+
\right]e_i.
\end{equation}
For each $i\geq1$, the one-dimensional Black--Scholes formula gives
\begin{equation}
e^{-r_i(T-t)}
\mathbb E^{\mathbb Q}
\left[
\left(X_T^{i,t,x_i}-K_i\right)^+
\right]
=
x_i\Phi(d_{1,i})
-
K_i e^{-r_i(T-t)}\Phi(d_{2,i}).
\end{equation}
Consequently, writing
\begin{equation}
V_i(t,x_i)
:=
x_i\Phi(d_{1,i})
-
K_i e^{-r_i(T-t)}\Phi(d_{2,i}),
\end{equation}
we have $V(t,x)=\sum_{i=1}^{\infty}V_i(t,x_i)e_i$.
Since $0\leq V_i(t,x_i)\leq x_i$,
\begin{equation}
\sum_{i=1}^{\infty}|V_i(t,x_i)|^2
\leq
\sum_{i=1}^{\infty}|x_i|^2
<\infty,
\end{equation}
and hence the series converges in $H$. For $t<T$, the standard Black--Scholes formulas yield
\begin{equation}
\partial_{x_i}V_i(t,x_i)=\Phi(d_{1,i}),
\qquad
\partial_{x_ix_i}V_i(t,x_i)
=
\frac{\varphi(d_{1,i})}
{x_i\sqrt{\lambda_i(T-t)}}.
\end{equation}
Since $A$ is bounded,
\begin{equation}
\sum_{i=1}^{\infty}
\left|
r_i x_i\partial_{x_i}V_i(t,x_i)
\right|^2
\leq
\sup_{i\geq1}|r_i|^2
\sum_{i=1}^{\infty}|x_i|^2
<\infty,
\end{equation}
and
\begin{equation}
\sum_{i=1}^{\infty}
\left|
r_iV_i(t,x_i)
\right|^2
\leq
\sup_{i\geq1}|r_i|^2
\sum_{i=1}^{\infty}|x_i|^2
<\infty.
\end{equation}
Moreover, since $\varphi(z)^2\leq(2\pi)^{-1}$,
\begin{equation}
\sum_{i=1}^{\infty}
\left|
\lambda_i x_i^2
\partial_{x_ix_i}V_i(t,x_i)
\right|^2
\leq
\frac{\sup_{i\geq1}\lambda_i}
{2\pi(T-t)}
\sum_{i=1}^{\infty}|x_i|^2
<\infty.
\end{equation}
All the above convergences are locally uniform in
$t\in[0,T)$. Finally, each $V_i$ satisfies the one-dimensional
Black--Scholes equation
\begin{equation}
\partial_tV_i
+
r_i x_i\partial_{x_i}V_i
+
\frac{\lambda_i}{2}x_i^2
\partial_{x_ix_i}V_i
-
r_iV_i
=
0.
\end{equation}
The preceding estimates allow these identities to be summed
in $H$, yielding \eqref{wbsh}.
\end{example}

The hedging problem for the asset process \eqref{gmbh} and any payoff whose value function satisfies \eqref{gbmbs} follows from \Cref{thm:replication-strong} and the discussion following it.  In the present setting, \Cref{prop:no_arbitrage} admits the following particular interpretation. If the portfolio strategy is also of diagonal form, as, for instance,
when hedging a diagonal payoff, then $V_t=\sum_{i=1}^{\infty}V_t^i e_i$ (where $V_t^i$ may depend only on $X_t^i$). In this case, \Cref{prop:no_arbitrage} implies that there cannot exist
$t\in(0,T]$ such that
\begin{equation}
V_t^i\geq0 \quad \text{a.s. for all } i\geq1,
\qquad
\mathbb P(V_t^i>0)>0
\quad \text{for some } i\geq1.
\end{equation}

In other words, arbitrage opportunities cannot arise in any individual market $i$. In the diagonal geometric Brownian motion setting, truncating the
infinite-dimensional Wiener process to its first $n$ coordinates yields
a finite-dimensional truncated market. Thus,
\Cref{lem:measurable_factorization} and the discussion following it show
that, under the corresponding injectivity condition, each truncated
market is complete, while the original infinite-dimensional market is
asymptotically complete.
\subsection{The Case of a Common Interest Rate $r_i=r$}
We now consider the case where there exists $r\in\mathbb R$ such that
\eqref{risknr} holds. Under $\mathbb Q$, the market is then given by
\begin{equation}
dX_t
=
rX_t\,dt
+
B(X_t)\,dW_t^{\mathbb Q}.
\end{equation}
In this case, the fundamental pricing formula reduces to
\begin{equation}
V(t,x)
=
e^{-r(T-t)}
\mathbb E^{\mathbb Q}
\left[
G\left(X_T^{t,x}\right)
\right].
\end{equation}
\begin{assumption} \label{assg2}
The function $G:H\to\mathbb R$ satisfies \cref{assg1} with the codomain
$H$ replaced by $\mathbb R$.
\end{assumption}
Using the same method as in \Cref{generi}, we obtain the following real-valued Black--Scholes equation. 
\begin{theorem}
Suppose that \Cref{assg2} holds. Then the value function
\begin{equation}
V(t,x)
=
e^{-r(T-t)}
\mathbb E^{\mathbb Q}
\left[
G\left(X_T^{t,x}\right)
\right]
\end{equation}
satisfies
\begin{equation}
\partial_t V(t,x)
+
rD_xV(t,x)x
+
\frac12
\sum_{i=1}^{\infty}
\lambda_i\langle x,e_i\rangle^2
\partial_{xx}V(t,x)(e_i,e_i)
-
rV(t,x)
=
0.
\end{equation}
\end{theorem}
We now consider two examples with tractable solutions.
\begin{example}
The first example is a diagonal-type European call option. Let
\begin{equation}
g(x)
=
\sum_{i=1}^{\infty}
a_i\bigl(\langle x,e_i\rangle-K_i\bigr)^+,
\qquad x\in H,
\end{equation}
where $a_i>0$, $(a_i)_{i\geq1}\in\ell^2$, and $K_i>0$. 
Writing $x_i=\langle x,e_i\rangle$ and $\tau=T-t$, the corresponding
option price is
\begin{equation}
V(t,x)
=
\sum_{i=1}^{\infty}
a_i
\left[
x_i\Phi\bigl(d_{1,i}(t,x)\bigr)
-
K_i e^{-r\tau}\Phi\bigl(d_{2,i}(t,x)\bigr)
\right],
\end{equation}
where
\begin{equation}
d_{1,i}(t,x)
=
\frac{
\log(x_i/K_i)
+
\left(r+\frac{\lambda_i}{2}\right)\tau
}{
\sqrt{\lambda_i\tau}
},
\qquad
d_{2,i}(t,x)
=
d_{1,i}(t,x)-\sqrt{\lambda_i\tau}.
\end{equation}
Thus, the infinite-dimensional option price is represented as a weighted
sum of one-dimensional Black--Scholes prices corresponding to the
coordinate processes. The verification that $V$ satisfies \eqref{wbsr} is the same as in \Cref{callop1}, applied
coordinatewise.
\end{example}
The second example considers an arithmetic basket option, obtained as the infinite-dimensional limit of finite-dimensional arithmetic basket options. In contrast to the diagonal case, the payoff depends jointly on all asset coordinates through their weighted sum, thereby capturing the aggregate behavior of the underlying assets.
\begin{example}
Consider the arithmetic basket payoff
\begin{equation}
g(x)
=
\left(
\sum_{i=1}^{\infty}a_i\langle x,e_i\rangle-K
\right)^+,
\qquad x\in H,
\end{equation}
where $a=(a_i)_{i\geq1}\in\ell^2$, $a_i >0$, and $K>0$. Writing
$x_i=\langle x,e_i\rangle$ and $\tau=T-t$, consider the finite-dimensional
approximations $g_N(x)=\left(\sum_{i=1}^{N}a_i x_i-K\right)^+$. The corresponding option prices are
\begin{equation}
V_N(t,x)
=
e^{-r\tau}
\mathbb E^{\mathbb Q}
\left[
\left(
\sum_{i=1}^{N}
a_i x_i
\exp\left(
\left(r-\frac{\lambda_i}{2}\right)\tau
+\sqrt{\lambda_i\tau}\,Z_i
\right)
-K
\right)^+
\right],
\end{equation}
where $Z_1,\ldots,Z_N$ are independent standard normal random variables.

Define
\begin{equation}
A_N
=
\sum_{i=2}^{N}
a_i x_i
\exp\left(
\left(r-\frac{\lambda_i}{2}\right)\tau
+\sqrt{\lambda_i\tau}\,Z_i
\right).
\end{equation}
Conditioning on $Z_2,\ldots,Z_N$ gives
\begin{equation}
V_N(t,x)
=
e^{-r\tau}
\int_{\mathbb R^{N-1}}
\Gamma\bigl(A_N(z_2,\ldots,z_N)\bigr)
\,\gamma_{N-1}(dz_2\cdots dz_N),
\end{equation}
where $\gamma_{N-1}$ denotes the standard Gaussian measure on
$\mathbb R^{N-1}$ and
\begin{equation}
\Gamma(a)
=
\begin{cases}
a_1x_1e^{r\tau}\Phi\bigl(d_1(a)\bigr)
-(K-a)\Phi\bigl(d_2(a)\bigr),
& a<K,\\[1mm]
a_1x_1e^{r\tau}+a-K,
& a\geq K,
\end{cases}
\end{equation}
with
\begin{equation}
d_1(a)
=
\frac{
\log\left(\dfrac{a_1x_1}{K-a}\right)
+
\left(r+\dfrac{\lambda_1}{2}\right)\tau
}{
\sqrt{\lambda_1\tau}
},
\qquad
d_2(a)
=
d_1(a)-\sqrt{\lambda_1\tau}.
\end{equation}
Equivalently, $V_N(t,x)=e^{-r\tau}
\mathbb E^{\mathbb Q}\left[\Gamma(A_N)\right]$.
Moreover,
\begin{equation}
A_N
\longrightarrow
A_\infty
:=
\sum_{i=2}^{\infty}
a_i x_i
\exp\left(
\left(r-\frac{\lambda_i}{2}\right)\tau
+\sqrt{\lambda_i\tau}\,Z_i
\right)
\end{equation}
in $L^1$. Since $\Gamma$ is Lipschitz continuous, it follows that
\begin{equation}
V(t,x)=
e^{-r\tau}
\mathbb E^{\mathbb Q}\left[\Gamma(A_\infty)\right]
=
e^{-r\tau}
\int_{\mathbb R}
\Gamma(a)\,\nu_{A_\infty}(da),
\end{equation}
where $\nu_{A_\infty}$ denotes the distribution of $A_\infty$. Thus, the
infinite-dimensional arithmetic basket option admits a semi-closed-form
representation as a one-dimensional integral with respect to the distribution
of $A_\infty$. It is straightforward to verify that the value function satisfies \eqref{wbsr}.
\end{example}
The relation between uniqueness of an equivalent martingale measure and market
completeness in markets with infinitely many traded assets has been studied in,
among others, \cite{jarrow1999hedging}. Our results are consistent with this
framework. In particular, \Cref{thm:rnm_gbm} establishes the existence and uniqueness of the risk-neutral measure, while the discussion following \Cref{lem:measurable_factorization} establishes the asymptotic completeness of the market.  More precisely, the asymptotic completeness obtained here through
approximation by finite-dimensional markets implies quasicompleteness
in the terminology of \cite{jarrow1999hedging}. In the present model,
the truncated asset-price processes are tradable, as truncation
corresponds to restricting trading to finitely many assets.
Thus, attainable claims are dense in the weak topology induced by
the class of signed local martingale measures. This is consistent
with the second fundamental theorem established in Theorem~5 of
\cite{jarrow1999hedging}, which relates uniqueness of the equivalent
martingale measure to quasicompleteness in an infinite-asset market.
\subsection{Financial Interpretation: Segmented Markets}

In this section, we provide a concrete two-dimensional example to illustrate a possible financial interpretation of our pricing and hedging framework in terms of segmented markets. For related studies on segmented financial markets, see, for instance, \cite{greenwood2018asset,goldstein2014speculation,kabanov2002hedging}. From a mathematical perspective, this example is simply a finite-dimensional truncation of the Hilbert-space-valued geometric Brownian motion introduced in Section~\ref{GBMsec}. Here, however, we assign an economic interpretation to its mathematical structure by viewing the individual components as distinct markets. Although we focus on two markets for simplicity, the construction can be readily extended to an arbitrary finite number of markets and, under suitable conditions, to the infinite-dimensional setting.

Consider two segmented markets, each equipped with its own risky asset and risk-free account, with possibly different interest rates. Let $H$ be a two-dimensional Hilbert space with an orthonormal basis $\{e_1,e_2\}$, and define $H_i=\operatorname{span}\{e_i\}$ for $i=1,2$, so that $H=H_1\oplus H_2$. We assume that trading across the two markets is prohibited, so that assets can only be traded within their respective markets. To reflect this restriction mathematically, we represent the monetary values associated with different markets along distinct orthogonal basis vectors. Specifically, each subspace $H_i$ represents a distinct market, with $e_i$ denoting its monetary unit. Thus, $ae_i$ represents an amount of $a$ units in market $i$.
Similarly, $be_2$ represents $b$ units of currency in the second market, but cannot represent a monetary value in the first market,
since $be_2\notin H_1$. Furthermore, no mapping between the spaces $H_i$ is introduced. The vector $ae_1+be_2$ then represents $a$ units of currency in the first market and $b$ units in the second market. Although these two components can be added as vectors in $H$ to represent a portfolio
of assets held in different markets, they cannot be directly combined into a single scalar monetary value without an additional conversion rule. In this way, the direct-sum structure provides a mathematical representation of market segmentation, while still allowing trading strategies in each market to depend on information from both markets.

Let $r_1,r_2\in\mathbb R$ be the constant risk-free interest rates of the two markets, with $r_1\neq r_2$. The corresponding bank accounts are given by
\begin{equation}
B_t^1=e^{r_1t},\qquad B_t^2=e^{r_2t},
\end{equation}
and the $H$-valued bank account process is $B_t=B_t^1e_1+B_t^2e_2$. Define $A\in L(H)$ by $Ae_i=r_ie_i$ for $i=1,2$. Then $S(t)=e^{tA}$ and $B_t=S(t)(e_1+e_2)$.

Let $X_t^1$ and $X_t^2$ denote the risky asset prices in the first and second markets, respectively. Under the risk-neutral measure $\mathbb Q$, suppose that
\begin{equation}
dX_t^1=r_1X_t^1\,dt+\sigma_1X_t^1\,dW_t^1,
\qquad
dX_t^2=r_2X_t^2\,dt+\sigma_2X_t^2\,dW_t^2,
\end{equation}
where $\sigma_1,\sigma_2>0$ and $W^1,W^2$ are independent standard Brownian motions under $\mathbb Q$. Equivalently, these dynamics can be expressed in the $H_i$-valued form as
\begin{equation}
dX_t^1e_1=r_1X_t^1e_1\,dt+\sigma_1X_t^1e_1\,dW_t^1,
\qquad
dX_t^2e_2=r_2X_t^2e_2\,dt+\sigma_2X_t^2e_2\,dW_t^2.
\end{equation}

Define the $H$-valued asset price process $X_t=X_t^1e_1+X_t^2e_2$. Its dynamics can be written as
\begin{equation}
dX_t=AX_t\,dt+\mathcal B(X_t)\,dW_t,
\end{equation}
where $W_t=W_t^1e_1+W_t^2e_2$ and $\mathcal B:H\to L(H)$ is defined by
\begin{equation}
\mathcal B(x)h
=\sigma_1\langle x,e_1\rangle_H\langle h,e_1\rangle_H e_1
+\sigma_2\langle x,e_2\rangle_H\langle h,e_2\rangle_H e_2.
\end{equation}

Consider a contingent claim with maturity $T$ and $H$-valued payoff
\begin{equation}
G(X_T)=G_1(X_T^1,X_T^2)e_1+G_2(X_T^1,X_T^2)e_2,
\end{equation}
where $G_1,G_2:\mathbb R^2\to\mathbb R$ are given payoff functions. Each component represents a payoff denominated in the monetary unit of its respective market, while both components may depend on the risky asset prices observed in the two markets.  Unlike the diagonal structure in \Cref{generi}, where each payoff component depends only on the asset price in its own market, the payoff components here may depend on asset prices from both markets. Consequently, although cross-market trading is prohibited, the two markets cannot be regarded as completely independent.

Each portfolio defined in \eqref{portde} admits a meaningful financial interpretation in this setting. Indeed, let $\Phi_t=(\phi_{ij}(t))_{i,j=1,2}$ and $\Psi_t=(\psi_{ij}(t))_{i,j=1,2}$ denote the operator-valued positions in the risky assets and bank accounts, respectively. The portfolio value can be expressed as
\begin{equation}
\begin{aligned}
V_t
={}&\Phi_tX_t+\Psi_tB_t\\
={}&\left(\phi_{11}(t)X_t^1+\phi_{12}(t)X_t^2+\psi_{11}(t)B_t^1+\psi_{12}(t)B_t^2\right)e_1\\
&+\left(\phi_{21}(t)X_t^1+\phi_{22}(t)X_t^2+\psi_{21}(t)B_t^1+\psi_{22}(t)B_t^2\right)e_2.
\end{aligned}
\end{equation}

In the first market, this portfolio can be interpreted as holding $\phi_{11}(t)$ units of the first risky asset and
\begin{equation}
\beta_t^1=\psi_{11}(t)+\frac{\phi_{12}(t)X_t^2+\psi_{12}(t)B_t^2}{B_t^1}
\end{equation}
units of the first bank account. Similarly, in the second market, the portfolio holds $\phi_{22}(t)$ units of the second risky asset and
\begin{equation}
\beta_t^2=\psi_{22}(t)+\frac{\phi_{21}(t)X_t^1+\psi_{21}(t)B_t^1}{B_t^2}
\end{equation}
units of the second bank account. Consequently,
\begin{equation}
V_t=\left(\phi_{11}(t)X_t^1+\beta_t^1B_t^1\right)e_1
+\left(\phi_{22}(t)X_t^2+\beta_t^2B_t^2\right)e_2.
\end{equation}

Thus, the portfolio can be interpreted entirely in terms of positions in the assets of the respective markets, without requiring any cross-market trading. The non-diagonal entries of $\Phi_t$ and $\Psi_t$ contribute to the local monetary values and may depend on information from the other market, but do not represent positions in assets traded across markets. In other words, although monetary values cannot be transferred between $H_1$ and $H_2$, the scalar quantities $X_t^i$ and $B_t^i$ may enter the portfolio value in either market, expressed in the monetary unit of the corresponding market.

We next consider the pricing of the contingent claim $G(X_T)$. By the risk-neutral pricing formula \eqref{prifo1}, its time-$t$ value is given by
\begin{equation}
\begin{aligned}
V_t
={}&S(t-T)\mathbb E^{\mathbb Q}[G(X_T)\mid\mathcal F_t]\\
={}&e^{-r_1(T-t)}\mathbb E^{\mathbb Q}[G_1(X_T^1,X_T^2)\mid\mathcal F_t]e_1\\
&+e^{-r_2(T-t)}\mathbb E^{\mathbb Q}[G_2(X_T^1,X_T^2)\mid\mathcal F_t]e_2.
\end{aligned}
\end{equation}

An interesting observation is that, when viewed from the perspective of an individual market, the pricing problem is closely related to that of a classical incomplete market. For instance, in the first market, only the risky asset $X^1$ and the bank account $B^1$ are traded, while $X^2$ can be regarded as an observable but non-traded risk factor. Since the payoff $G_1(X_T^1,X_T^2)$ may depend on both sources of randomness, it is generally not attainable by trading solely in the first market. Nevertheless, its price under the chosen equivalent martingale measure $\mathbb Q$ is given by
\begin{equation}
V_t^1=e^{-r_1(T-t)}\mathbb E^{\mathbb Q}[G_1(X_T^1,X_T^2)\mid\mathcal F_t],
\end{equation}
which is precisely the risk-neutral valuation formula associated with that measure in an incomplete market. The same observation applies to the second market. Thus, the present framework is consistent with a familiar interpretation of incomplete-market pricing at the level of individual market values, while retaining the $H$-valued portfolio formulation introduced above.

We now construct a replicating portfolio following the approach of \Cref{thm:replication-strong}. Write the pricing function as $V(t,x)=V_1(t,x)e_1+V_2(t,x)e_2$, where $x=x_1e_1+x_2e_2$. Assuming sufficient regularity, define the operator-valued risky asset position by
\begin{equation}
\Phi_t=D_xV(t,X_t)
=
\begin{pmatrix}
\partial_{x_1}V_1(t,X_t)&\partial_{x_2}V_1(t,X_t)\\
\partial_{x_1}V_2(t,X_t)&\partial_{x_2}V_2(t,X_t)
\end{pmatrix}.
\end{equation}

To determine the bank account position $\Psi_t$, set
\begin{equation}
\begin{aligned}
R_t&=V(t,X_t)-D_xV(t,X_t)X_t,\\
Q_t&=\partial_tV(t,X_t)
+\frac12\operatorname{Tr}_{\mathcal B(X_t)}\bigl(D_x^2V(t,X_t)\bigr).
\end{aligned}
\end{equation}
As in \Cref{thm:replication-strong}, we choose $\Psi_t\in L(H)$ such that
\begin{equation}
\Psi_tB_t=R_t,\qquad \Psi_tAB_t=Q_t.
\end{equation}
Since $r_1\neq r_2$, the vectors $B_t$ and $AB_t$ are linearly independent, so these conditions uniquely determine $\Psi_t$. In particular, writing $R_t=R_t^1e_1+R_t^2e_2$ and $Q_t=Q_t^1e_1+Q_t^2e_2$, we obtain
\begin{equation}
\Psi_t=
\begin{pmatrix}
\dfrac{Q_t^1-r_2R_t^1}{(r_1-r_2)B_t^1}
&
\dfrac{r_1R_t^1-Q_t^1}{(r_1-r_2)B_t^2}\\[8pt]
\dfrac{Q_t^2-r_2R_t^2}{(r_1-r_2)B_t^1}
&
\dfrac{r_1R_t^2-Q_t^2}{(r_1-r_2)B_t^2}
\end{pmatrix}.
\end{equation}

By the pricing equation, $Q_t=AV(t,X_t)-D_xV(t,X_t)AX_t$. Consequently, the portfolio satisfies
\begin{equation}
V(t,X_t)=\Phi_tX_t+\Psi_tB_t,\qquad
dV(t,X_t)=\Phi_t\,dX_t+\Psi_t\,dB_t.
\end{equation}
Moreover, $S(-t)V(t,X_t)$ is a $\mathbb Q$-martingale, and $V(T,X_T)=G(X_T)$. Hence, the portfolio is admissible in the sense of Definition~\ref{selfstrong} and replicates the claim.

We can now express the portfolio in terms of the actual asset holdings in each market. In the first market, the portfolio holds $\partial_{x_1}V_1(t,X_t)$ units of the first risky asset and
\begin{equation}
\begin{aligned}
\beta_t^1
&=\psi_{11}(t)+\frac{\phi_{12}(t)X_t^2+\psi_{12}(t)B_t^2}{B_t^1}\\
&=\frac{V_1(t,X_t)-X_t^1\partial_{x_1}V_1(t,X_t)}{B_t^1}
\end{aligned}
\end{equation}
units of the first bank account. Similarly, in the second market, it holds $\partial_{x_2}V_2(t,X_t)$ units of the second risky asset and
\begin{equation}
\begin{aligned}
\beta_t^2
&=\psi_{22}(t)+\frac{\phi_{21}(t)X_t^1+\psi_{21}(t)B_t^1}{B_t^2}\\
&=\frac{V_2(t,X_t)-X_t^2\partial_{x_2}V_2(t,X_t)}{B_t^2}
\end{aligned}
\end{equation}
units of the second bank account.

Thus, the replicating portfolio admits the following market-by-market interpretation:
\begin{equation}
\begin{aligned}
V(t,X_t)
={}&\left(\partial_{x_1}V_1(t,X_t)X_t^1+\beta_t^1B_t^1\right)e_1\\
&+\left(\partial_{x_2}V_2(t,X_t)X_t^2+\beta_t^2B_t^2\right)e_2.
\end{aligned}
\end{equation}
Although the positions in each market may depend on the asset prices of both markets, the portfolio can be interpreted using only the risky asset and bank account belonging to each respective market. Note that, under (i) of \Cref{selfstrong}, the self-financing condition is imposed on the aggregate portfolio $V$, not separately on each market component, so the decomposition above describes the holdings in each market at each time rather than a self-financing strategy within each market.

\bibliographystyle{elsarticle-num-names}
\bibliography{sample}

@book{da2014stochastic,
  title={Stochastic equations in infinite dimensions},
  author={Da Prato, Giuseppe and Zabczyk, Jerzy},
  year={2014},
  publisher={Cambridge university press}
}

@book{hytonen2016analysis,
  title={Analysis in Banach spaces},
  author={Hyt{\"o}nen, Tuomas and Van Neerven, Jan and Veraar, Mark and Weis, Lutz},
  volume={63},
  year={2016},
  publisher={Springer}
}

@book{benth2023stochastic,
  title={Stochastic Models for Prices Dynamics in Energy and Commodity Markets: An Infinite-Dimensional Perspective},
  author={Benth, Fred Espen and Kr{\"u}hner, Paul},
  year={2023},
  publisher={Springer Nature}
}

@article{cont2021,
  title={A stochastic partial differential equation model for limit order book dynamics},
  author={Cont, Rama and M{\"u}ller, Marvin S},
  journal={SIAM Journal on Financial Mathematics},
  volume={12},
  number={2},
  pages={744--787},
  year={2021},
  publisher={SIAM}
}

@book{meyer2012banach,
  title={Banach lattices},
  author={Meyer-Nieberg, Peter},
  year={2012},
  publisher={Springer Science \& Business Media}
}

@article{desoer1963note,
  title={A note on pseudoinverses},
  author={Desoer, Charles A and Whalen, Barry H},
  journal={Journal of the Society for Industrial and Applied Mathematics},
  volume={11},
  number={2},
  pages={442--447},
  year={1963},
  publisher={SIAM}
}

@article{jarrow1999hedging,
  title={Hedging contingent claims on semimartingales},
  author={Jarrow, Robert and Madan, Dilip B},
  journal={Finance and Stochastics},
  volume={3},
  number={1},
  pages={111--134},
  year={1999},
  publisher={Springer}
}

@book{filipovic2001consistency,
  title={Consistency problems for Heath-Jarrow-Morton interest rate models},
  author={Filipovic, Damir},
  year={2001},
  publisher={Springer Science \& Business Media}
}

@article{heath1992bond,
  title={Bond pricing and the term structure of interest rates: A new methodology for contingent claims valuation},
  author={Heath, David and Jarrow, Robert and Morton, Andrew},
  journal={Econometrica: Journal of the Econometric Society},
  pages={77--105},
  year={1992},
  publisher={JSTOR}
}

@article{filipovic2010term,
  title={Term structure models driven by Wiener processes and Poisson measures: existence and positivity},
  author={Filipovi{\'c}, Damir and Tappe, Stefan and Teichmann, Josef},
  journal={SIAM Journal on Financial Mathematics},
  volume={1},
  number={1},
  pages={523--554},
  year={2010},
  publisher={SIAM}
}

@book{carmona2006interest,
  title={Interest rate models: an infinite dimensional stochastic analysis perspective},
  author={Carmona, Ren{\'e} A and Tehranchi, Michael R},
  year={2006},
  publisher={Springer}
}

@article{benth2014representation,
  title={Representation of infinite-dimensional forward price models in commodity markets},
  author={Benth, Fred Espen and Kr{\"u}hner, Paul},
  journal={Communications in Mathematics and Statistics},
  volume={2},
  number={1},
  pages={47--106},
  year={2014},
  publisher={Springer}
}

@article{benth2015derivatives,
  title={Derivatives pricing in energy markets: an infinite-dimensional approach},
  author={Benth, Fred Espen and Kr{\"u}hner, Paul},
  journal={SIAM Journal on Financial Mathematics},
  volume={6},
  number={1},
  pages={825--869},
  year={2015},
  publisher={SIAM}
}

@article{taflin2005bond,
  title={Bond market completeness and attainable contingent claims},
  author={Taflin, Erik},
  journal={Finance and Stochastics},
  volume={9},
  number={3},
  pages={429--452},
  year={2005},
  publisher={Springer}
}

@article{aihara2005stochastic,
  title={Stochastic Hyperbolic Dynamics For Infinite-Dimensional Forward Rates And Option Pricing},
  author={Aihara, Shin Ichi and Bagchi, Arunabha},
  journal={Mathematical Finance: An International Journal of Mathematics, Statistics and Financial Economics},
  volume={15},
  number={1},
  pages={27--47},
  year={2005},
  publisher={Wiley Online Library}
}

@article{chang1999european,
  title={The European option with hereditary price structures: basic theory},
  author={Chang, Mou-Hsiung and Youree, Roger K},
  journal={Applied Mathematics and Computation},
  volume={102},
  number={2-3},
  pages={279--296},
  year={1999},
  publisher={Elsevier}
}

@article{chang2007infinite,
  title={Infinite-dimensional Black-Scholes equation with hereditary structure},
  author={Chang, Mou-Hsiung and Youree, Roger K},
  journal={Applied Mathematics and Optimization},
  volume={56},
  number={3},
  pages={395--424},
  year={2007},
  publisher={Springer}
}

@article{arriojas2007delayed,
  title={A delayed Black and Scholes formula},
  author={Arriojas, Mercedes and Hu, Yaozhong and Mohammed, Salah-Eldin and Pap, Gyula},
  journal={Stochastic Analysis and Applications},
  volume={25},
  number={2},
  pages={471--492},
  year={2007},
  publisher={Taylor \& Francis}
}

@inproceedings{da2007kolmogorov,
  title={Kolmogorov equations for stochastic PDE's with multiplicative noise},
  author={Da Prato, Giuseppe},
  booktitle={Stochastic Analysis and Applications: The Abel Symposium 2005},
  pages={235--263},
  year={2007},
  organization={Springer}
}

@book{diestel1977vector,
  author    = {Diestel, Joseph and Uhl, J. Jerry},
  title     = {Vector Measures},
  series    = {Mathematical Surveys},
  volume    = {15},
  publisher = {American Mathematical Society},
  address   = {Providence, RI},
  year      = {1977},
  doi       = {10.1090/surv/015},
  isbn      = {9780821815151}
}

@article{greenwood2018asset,
  title={Asset price dynamics in partially segmented markets},
  author={Greenwood, Robin and Hanson, Samuel G and Liao, Gordon Y},
  journal={The Review of Financial Studies},
  volume={31},
  number={9},
  pages={3307--3343},
  year={2018},
  publisher={Oxford University Press}
}

@article{goldstein2014speculation,
  title={Speculation and hedging in segmented markets},
  author={Goldstein, Itay and Li, Yan and Yang, Liyan},
  journal={The Review of Financial Studies},
  volume={27},
  number={3},
  pages={881--922},
  year={2014},
  publisher={Oxford University Press}
}

@article{kabanov2002hedging,
  title={Hedging under transaction costs in currency markets: A continuous-time model},
  author={Kabanov, Yuri M and Last, G{\"u}nter},
  journal={Mathematical Finance},
  volume={12},
  number={1},
  pages={63--70},
  year={2002},
  publisher={Wiley Online Library}
}

@book{pazy2012semigroups,
  title={Semigroups of linear operators and applications to partial differential equations},
  author={Pazy, Amnon},
  year={2012},
  publisher={Springer Science \& Business Media}
}

@article{pronk2014tools,
  title={Tools for Malliavin calculus in UMD Banach spaces},
  author={Pronk, Matthijs and Veraar, Mark},
  journal={Potential Analysis},
  volume={40},
  number={4},
  pages={307--344},
  year={2014},
  publisher={Springer}
}

\end{document}